\documentclass[11.5pt]{article}

\usepackage{amssymb,amsmath,amsthm,bbm,enumitem,xcolor}
\usepackage{eurosym}
\usepackage[a4paper, total={6.3in, 9.2in}]{geometry}
\usepackage{setspace}
\usepackage{hyperref}
\usepackage{authblk}
\usepackage{graphicx,caption,subcaption}
\usepackage[title]{appendix}
 \usepackage{mathtools}
\usepackage{algorithm,algpseudocode}
\usepackage{tikz}
\usetikzlibrary{shapes.geometric, arrows, decorations.markings,shapes.multipart,cd,positioning,patterns,snakes}
\tikzstyle{startstop} = [rectangle, rounded corners, minimum width=2cm, minimum height=1cm,text centered, draw=black]
\tikzstyle{vecArrow} = [thick, decoration={markings,mark=at position
   1 with {\arrow[semithick]{open triangle 60}}},
   double distance=1.4pt, shorten >= 5.5pt,
   preaction = {decorate},
   postaction = {draw,line width=1.4pt, white,shorten >= 4.5pt}]
\tikzstyle{innerWhite} = [semithick, white,line width=1.4pt, shorten >= 4.5pt]
   \tikzset{|/.tip={Bar[width=.8ex,round]}}

\usepackage{apacite} 
\AtBeginDocument{%

   \urlstyle{APACsame}%
}

\usepackage[round]{natbib}

\theoremstyle{definition}
\newtheorem{definition}{Definition}[section]
\theoremstyle{definition}
\newtheorem{example}{Example}[section]
\theoremstyle{definition}

\theoremstyle{definition}
\newtheorem{proposition}{Proposition}[section]
\theoremstyle{definition}

\theoremstyle{remark}
\newtheorem{remark}{Remark}[section]
\theoremstyle{definition}
\newtheorem{lemma}{Lemma}[section]

\theoremstyle{definition}
\newtheorem{problem}{Problem}

\setlist[itemize]{leftmargin=0.4cm,labelindent=\parindent}

\title{Capital-Allocation-Induced Risk Sharing}
\date{\today}

\author[$\dagger$]{Wing Fung Chong}
\author[$\star$]{Runhuan Feng}
\author[$\star\star$]{Kenneth Tsz Hin Ng}

\affil[$\dagger$]{Department of Statistics and Actuarial Science, School of Computing and Data Science, The University of Hong Kong. Email: chongwf@hku.hk.}
\affil[$\star$]{School of Economics and Management, Tsinghua University. Email: fengrh@sem.tsinghua.edu.cn.}
\affil[$\star\star$]{Department of Mathematics, The Ohio State University. Email: ng.499@osu.edu.}

\begin{document}
\maketitle

\begin{abstract} 
This article proposes a new class of risk-sharing rules by exploring the relationship between capital allocation and risk sharing. While the former is concerned with \textit{ex-ante} allocating capitals to different lines of business within a corporation based on the relationship among the individual risks, often also through the aggregate risk, the latter is an arrangement which collects risks from and allocates them to, also \textit{ex-ante}, a group of participants. Drawing on this analogy, we introduce a novel idea of inducing risk-sharing rules by randomizing existing capital allocation principles. Such an approach derives new risk-sharing rules complementing known results in the literature, which were largely based on economic principles and Pareto optimality.
\end{abstract}

\begin{flushleft}

\end{flushleft}

\section{Introduction}

Risk sharing is prevalent in many industries, including supply chain risk pooling, collaborative logistics, healthcare resource allocation, financial risk management, and project delivery partnerships. A pension fund manager preparing to invest in a significant risk transfer transaction, like the \euro 2.1 billion corporate loan portfolio recently shared between PGGM and UniCredit Bulbank\footnote{\url{https://europeanpensions.net/ep/Dutch-PGGM-and-UniCredit-Bulbank-strike-2-1bn-risk-sharing-deal.php}}, must decide how to distribute credit risk among institutional investors, each with different risk tolerance and capital constraint. A construction firm entering an integrated project delivery agreement, where the owner, designer, contractor, and subcontractor must agree upfront on how to share cost overruns when they arise; see \cite{eissa}. These are daily operational challenges faced by organizations managing complex risks under uncertainty.

While these risk sharing problems may appear distinct from the challenges faced by a bank's Chief Risk Officer allocating economic capital across trading desks, or an insurance executive determining how much capital each business line should hold to satisfy regulators, yet, beneath the surface, the problems of risk sharing and capital allocation share a similar mathematical structure. Both involve a group of entities, participants in a risk pool, lines of business (LOBs) within a firm, or partners in a joint venture, who must agree on an allocation of uncertain future outcomes. In capital allocation, the allocation is deterministic: the firm decides today how much capital each unit must hold against future losses. In risk sharing, the allocation is random: participants agree today on rules that determine how actual losses will be distributed when they occur.
This paper reveals that these two techniques on managing risk, capital allocation and risk sharing, are not merely analogous but are mathematically dual in the following sense. We shall show that every class of capital allocation principles can be transformed into a risk-sharing rule through a novel randomization approach; conversely, many well-known risk-sharing rules can be interpreted as randomized capital allocations principles.

    \subsection{Capital Allocation}
    Capital allocation is a procedure for a corporate to allocate a capital to each of its LOBs.
    Mathematically, let ${\bf X}=(X_i)_{i=1}^n$ be a random risk vector where $X_i$, $i=1,2,\dots,n$, is the random loss of the $i$-th LOB, and let $K_i\left({\bf X}\right)$ be the {\it ex-ante} allocated risk capital to this LOB. A capital allocation principle is a mapping

        \begin{equation*}
            {\bf X} = (X_1,\dots,X_n) \mapsto {\bf K}\left({\bf X}\right) = (K_1\left({\bf X}\right),\dots,K_n\left({\bf X}\right)) \in \mathbb{R}^n,
        \end{equation*}  
    such that $\sum_{i=1}^nK_i\left({\bf X}\right) = K\left({\bf X}\right)$.  Depending on how the aggregate capital $K\left({\bf X}\right)$ is determined, capital allocation can mainly be classified into top-down and bottom-up principles.
      \begin{itemize}

        \item {\bf Top-down principle}. The total capital $K\left({\bf X}\right)$ 
        is first determined, for example, by some risk measure of the aggregate risk $S=\sum_{i=1}^{n}X_i$, which is subsequently allocated to the LOBs.
        
            \item {\bf Bottom-up principle}. The capitals $K_i\left({\bf X}\right)$, $i=1,2,\dots,n$, of the LOBs are first determined, which are subsequently summed up to the total capital $K\left({\bf X}\right)$ for the corporate.

            
        \end{itemize}

\noindent
        The main purposes of capital allocation are twofold. First, it protects the LOBs from default in case of detrimental scenarios. Second, it provides a useful reference for the corporate to measure the performance among different LOBs.
        
        As mentioned in \cite{merton:perold:1993}, capital allocation incentivizes diversifying the risks of the LOBs, since the marginal increase of the risk of a LOB, when determined as an individual firm, as opposed to as a part of a bigger corporate, are generally different. They also noted that the tempting approach of simply allocating the standalone amount of capital to all LOBs is infeasible, since the firm’s total required capital is very often smaller than the sum of the standalone capital requirements of the individual LOBs. Moreover, such an approach would jeopardize the return on capital of each LOB, since the allocated capital would exceed the amount justified by its actual contribution to the firm’s overall risk. This inefficiency can also be interpreted through the frictional cost of capital, defined as the discrepancy between the target profit, which is the profit required to earn the cost of capital on the allocated amount, and the expected profit of the LOB; see, e.g., \cite{bauer:2016} and the references therein. Over-allocation of capital increases this frictional cost, as business lines must generate higher profits merely to meet the required return on excessive capital holdings. Hence, capital allocation can also be viewed as a cost allocation procedure, and a naive standalone allocation to all LOBs is clearly not efficient.

    \subsection{Risk Sharing}
    Risk-sharing is a  practice of redistributing risks among a group of participants, where each of them brings their own risk into a sharing pool and agrees to a specific sharing rule \textit{ex-ante}.  The aggregate loss in the pool is subsequently distributed  among the participants after a detrimental event according to the collectively agreed sharing rule. Therefore, unlike capital allocation, risk-sharing is random in nature, which is determined \textit{ex-ante}, but realized \textit{ex-post}. Mathematically, let ${\bf X}=(X_i)_{i=1}^n$ be a random risk vector where $X_i$, $i=1,2,\dots,n$, is the risk suffered by the $i$-th participant. A risk-sharing rule is a mapping 
        \begin{equation*}
            {\bf X} = (X_1,\dots,X_n) \mapsto {\bf H}({\bf X}) = (H_1({\bf X}),\dots H_n({\bf X})),
        \end{equation*}
    such that
     \begin{equation}
     \label{eq:sum:preserve}
            \sum_{i=1}^n H_i({\bf X}) = S =\sum_{i=1}^{n}X_i, \ \text{almost surely}.
        \end{equation}
   In other words, the $i$-th participant of the pool has to contribute $H_i({\bf x})$ \textit{ex-post} in exchange for the realized loss $X_i=x_i$ that she suffered, where ${\bf x} = (x_1,\dots,x_n)$, such that the aggregate loss is exactly covered by the sharing pool. We shall occasionally call such property \eqref{eq:sum:preserve} as the sum-preserving property in the sequel.
   
   Ideally, a practical risk-sharing rule should provide sufficient incentives to participants so that each of them is entertained and agrees to join the pool at the first place. One fundamental approach is to construct rules that are Pareto-optimal; see, e.g. the seminal work by \cite{borch1992equilibrium}. Others have taken a more axiomatic and/or economic approach and proposed a list of desirable properties that a practical risk-sharing rule should satisfy; see, for example,  \cite{DENUIT2021116,hieber_lucas_2022}. 


    

 \subsection{Contributions and Organizations}
    
    Clearly, both capital allocation and risk sharing utilize diversification. Despite being different in the nature of the allocations, which are deterministic capitals for capital allocation, and random risks for risk-sharing, the similarities between these two strategies have been long identified, at least implicitly. Notably, \cite{DENUIT2021116} presented an example that subtly connects the Euler principle in capital allocation with the conditional mean risk-sharing (CMRS) and the quantile risk-sharing rules. In some occasions, the two terminologies may even be used interchangeably; see, e.g., \cite{filipovic2008optimal}. However, to our best knowledge, there is not much systematic study of the connections between the two strategies in the literature, except a recent paper by \cite{BOONEN2025133}, which linked them by the so-called ${\bf K}$-fair property. In this paper, we bridge the two risk management strategies by proposing a novel randomization approach, which constructs risk-sharing rules from  capital allocation principles.  
    
    The idea of the randomization approach can be outlined as follows. First, we confine ourselves to a rich enough family of capital allocation principles. Next, for each possible realization of the risks, we identify a member allocation principle from the family whose aggregate capital agrees with the realized aggregate loss in that scenario. For example, the members of the family can be given by the top-down Euler principle of allocations based on the same risk measure, but the parametric level of the risk measure differs among member allocation principles; then, find an appropriate parametric level of the risk measure on the aggregate risk which agrees with the realized aggregate loss. This provides us a sampling rule from the family of capital allocation principles and eventually leads to a risk-sharing rule when all possible realizations are assembled. Hence, this method is a scenario-wise approach where losses are allocated according to a specific capital allocation principle within the family for each possible realization. 
    
    
    The characteristics of this approach are two-folds. First, it deviates from conventional approaches of risk-sharing, which are mostly based on certain economic principles and Pareto optimality, and provides an alternative systematic and unified framework of inducing risk-sharing rules.  This also provides new perspectives when interpreting some existing risk-sharing rules. Second, by definition, the induced sharing rule inherits the motivations and properties of its parent family of capital allocation principles in a scenario-wise manner. Specifically, if a particular property is satisfied by every member in the family, it will be directly inherited by the induced risk-sharing rule such that it holds for every possible realization. For example, if we consider a family of Pareto-optimal capital allocation principles, the induced risk-sharing rule will naturally be Pareto-optimal for {\it each} possible realization. This is different from the usual notion of Pareto optimality in the risk-sharing literature, which is a distributional property, whereas the Pareto optimality induced herein is a pointwise property and is satisfied for each possible realization. 

    This paper is organized as follows. In Section \ref{sec:math}, we first provide a provoking example that motivates the idea of inducing risk-sharing rules from capital allocation principles. The generic mathematical formulation of the randomization approach are then provided in Definition \ref{def:induce:risk:sharing}. Specifically, an induced risk-sharing rule is characterized by an underlying family of capital allocation principles along with a parameter sampling rule. To induce a legitimate risk-sharing rule, the randomized aggregate capital needs to match with the aggregate risk in all possible realizations so that the sum-preserving property \eqref{eq:sum:preserve} is satisfied.  This can be achieved if the aggregate capital function is surjective onto the support of the aggregate risk $S$, and admits a measurable right-inverse, a property that we shall refer to as measurable right-invertibility. In Section \ref{sec:risk:aggregation}, we discuss this requirement in detail, with a focus on  this property for top-down principles.
    

In Section \ref{sec:top:down}, we study the induction of risk-sharing rules from top-down allocation principles, where the parameter sampling rule is determined entirely by the chosen aggregate capital method. We place particular emphasis on the optimization and Euler principles, deriving closed-form expressions for the induced sharing rules in certain cases, such as for specific allocation principles or when the underlying risks follow an elliptical distribution. We further highlight the scenario-wise property of the induced risk-sharing rules inherited from the parent capital allocation principles. Several existing sharing rules can be recovered within this top-down framework. For instance, the quantile risk-sharing rule emerges by randomizing the optimization principle with an absolute deviation penalty, providing a new perspective of that rule.

A major focus is on cases where the aggregate capital function is defined by a family of distortion risk measures. For this setting, we develop in Proposition \ref{pp:distortion} sufficient conditions on the underlying family of distortion functions that ensure the measurable right-invertibility of the aggregate capital function over the support  of $S$. We also connect the comonotonicity of the Euler-induced sharing rule with that of the CMRS rule; see Proposition \ref{pp:distortion:euler:comono}. 

In Section \ref{sec:bottom-up}, we turn to bottom-up principles, where the parameter sampling rule is determined by the underlying capital allocation method, as the aggregate capital emerges endogenously. In particular, we study the weighted risk allocation principle introduced in \cite{FURMAN2008263}, and the holistic principle introduced in \cite{chong2021holistic}.  In Propositions \ref{pp:weighted:surjective} and \ref{pp:holistic}, we establish the measurable right-invertibility of the endogenous aggregate capital functions on the support of $S$ for the two principles, respectively.
We further derive the corresponding sharing rule for an elliptically distributed risk vector, discussing how this rule differs from that obtained under the top-down optimization principle. The paper is finally concluded in Section \ref{sec:conclusion}.

    \subsection{Related Literature} 
    Capital allocation is the backbone of the randomization approach introduced in this paper. Numerous allocation principles have been developed in the literature based on different economic motivations.  For top-down principles, \cite{tasche2004allocating} introduced the gradient approach by the Euler principle where allocations are based on the marginal effect of a LOB towards the aggregate risk. \cite{Pesenti:dynamic:risk:2025} generalized the Euler principle with dynamic risk contributions. \cite{bauer:2016} considered a reverse approach which aims at finding the suitable risk measure that would lead to a target marginal cost of risk derived based on the economic environment. The optimization approach in \cite{LAEVEN2004299,dhaene}
    seeks for an allocation principle such that the discrepancies between the individual risks and the allocated capitals are minimized. For axiomatic approaches considered in \cite{denault2001coherent,kalkbrener:2005}, desirable properties of a capital allocation principle are first proposed, and principles of allocations meeting the properties are subsequently deduced. For game-theoretic approaches (see, e.g., \cite{denault2001coherent,TSANAKAS2003239,BOONEN201795}), the objective is formulated as a cost allocation in cooperative games whose solution is given by the Aumann-Shapley value. Another closely related concept is risk decomposition (see, e.g.~\cite{schilling:decomposing:2020}), which systematically attributes aggregate risk to components,
    thereby informing capital allocation decisions.
    
    As  examples of bottom-up principles, \cite{FURMAN2008263} introduced the weighted risk capital allocation,
    subsequently, \cite{Furman03072021} combined the bottom-up and top-down principles into a single unifying approach. \cite{chong2021holistic} proposed a holistic principle where capital aggregation and allocation are solved simultaneously in a single optimization problem;
    in \cite{CHEN2021359}, such holistic principle was utilized for managing resources for pandemic risk.
    
    For a more comprehensive survey on capital allocation, consult \cite{guo2021capital} and the references therein. Notice that different approaches may end up leading to the same allocation principle. For instance, if the risk measure in the game-theoretic approach satisfies some additional properties, the Aumann-Shapley value coincides with the allocated capital based on the Euler principle.


    The primitive formulations of risk-sharing mechanisms can be dated back to the seminal works of \cite{arrow1954existence}, \cite{borch1960safety} and \cite{borch1992equilibrium}, which respectively studied the equilibrium in competitive economy and reinsurance markets. In recent years, research on risk-sharing has seen fruitful development, driven in part by its growing applications, for instance, 
\cite{ABDIKERIMOVA2022735,feng2022peer,abdikerimova2024multiperiod,ghossoub2024pareto,anthropelos2026expansion}. Over decades of development, achieving Pareto optimality has been generally acknowledged as a major principle of a risk-sharing mechanism, which is largely connected to the comonotonicity of the allocated risk; see, e.g., \cite{landsberger1994co,LUDKOVSKI20081181,CARLIER2012207,liu:wang:inf-con:2002}.
    
    Choosing risk measures with desirable properties plays a prominent role in constructing Pareto-optimal risk-sharing. Apart from the popular Value-at-Risk (VaR), considerable attentions have also be given to coherent risk measures in \cite{Artzner:coherent:1999}, and convex risk measures in \cite{follmer2002convex}
    such as the Expected Shortfall (ES), which was shown to be uniquely characterized by desirable economic axioms in \cite{wang:axiom:ES:2021}. Entropy-based coherent and convex risk measures were also studied in \cite{roger:entropy:risk:measure}. As a generalization of VaR and ES, \cite{embrechts:quantile:2018} introduced the Range VaR, and the notion of robustness in risk-sharing. Subsequently, \cite{castagnoli:star:shape:2022} developed the theory of star-shaped risk measures, which weakens the convexity requirement and includes the class of all convex risk measures as well as VaR. In a related contribution, \cite{embrechts:robustness:2022} analyzed robustness properties of risk-measure-based optimization problems, with particular attention to VaR and convex risk measures; \cite{WANG2023322} introduced a scenario-dependent shortfall risk measure with applications in capital allocations; \cite{liu:lambda:VaR:2025} studied risk-sharing problems using Lambda Value at Risk.
    
    Our proposed method sets apart from the aforementioned risk-sharing works. Instead of considering Pareto-optimal sharing in a distributional sense, the randomization approach introduced in this paper allows us to consider a family of capital allocation principles where each member is Pareto-optimal for each realization, for instance, via the optimization principle by \cite{dhaene} and the holistic principle by \cite{chong2021holistic}.
    Furthermore, the resulting shared-risks do not need to be comonotonic to achieve this scenario-wise Pareto optimality.
    


  Apart from being motivated by the Pareto optimality, another popular risk-sharing approach  is the CMRS rule introduced in \cite{DENUIT2012265}. The CMRS rule and its applications were further studied in \cite{denuit_2019}, \cite{denuit2020conditional}, \cite{DENUIT2021116}, \cite{axiomatic:anonymized}.
  \cite{dhaene:2021} introduced the quantile risk-sharing rule, which is equivalent to the comonotonic CMRS rule; \cite{denuit2024conditional} considered the diversification effect of risks in a sharing pool by the CMRS rule.  In this paper, we shall show that the quantile risk sharing and the CMRS rules can be derived by the randomization approach from certain capital allocation principles, which further reinforces their connotations.



\subsection{Notation} 
Consider a probability space $(\Omega,\mathcal{F},\mathbb{P})$, with $\mathbb{E}$ being the expectation operator taken with respect to $\mathbb{P}$. We also denote by $\mathbb{E}^\mathbb{Q}$ the expectation operator taken with respect to a generic probability measure $\mathbb{Q}$ on $(\Omega,\mathcal{F})$. We write $\mathbb{P}\sim \mathbb{Q}$, if $\mathbb{P}$ and $\mathbb{Q}$ are equivalent probability measures. For $p\geq 1$, the space of all $L^p$-integrable random variables with respect to $\mathbb{Q}$ is denoted by $L^p(\Omega,\mathcal{F},\mathbb{Q})$. 

For any real-valued random variable $X$, we let $F_X$, $\bar{F}_X$ and $F_X^{-1}$ be the  cumulative distribution function, the survival function, and the (left inverse) quantile function of $X$, respectively; we use the VaR notation, $\text{VaR}_\cdot(X)$, interchangeably with $F_X^{-1}(\cdot)$. We also consider the right inverse $F_X^{-1+}$ of $F_X$, which is defined as 
    \begin{equation*}
        F_X^{-1+}(p) := \sup\{x \in\mathbb{R} : F_X(x)\leq p \}, \ p\in[0,1].
    \end{equation*}

 Let ${\bf X}=(X_i)_{i=1}^n$ be an $n$-dimensional random risk vector and define the aggregate risk by $S=\sum_{i=1}^n X_i$, where $S$ is finite almost surely. Denote the support of $S$ by $\mathcal{S}_{\bf X}$, which satisfies $\mathcal{S}_{\bf X}\subset [F_S^{-1+}(0),F_S^{-1}(1)]$. Here, we take the convention that for any $a,b\in\mathbb{R}\cup\{\pm\infty\}$, $a<b$, $[a,b]:=\{x\in\mathbb{R}: a\leq x\leq b\}$, so that $[a,b]\subset\mathbb{R}$. For $i,j=1,2,\dots,n$, we also write $\mu_i := \mathbb{E}[X_i]$, $\sigma_i^2:=\text{Var}[X_i]$, $\sigma_{i,j}:= \text{Cov}[X_i,X_j]$, $\mu_S:=\mathbb
 E[S]$, $\sigma_S^2:=\text{Var}[S]$ and $\sigma_{i,S} := \text{Cov}[X_i,S]$.
 

\section{Randomization of Capital Allocation Principles}
\label{sec:math}

    Capital allocation and risk sharing are both risk management strategies that allocate resources or risks based on the risk vector ${\bf X}$. Table \ref{tab:compare:CA:RS} compares them in different aspects.
    
     
    \begin{table}[!h]
        \centering
        \begin{tabular}{|c|c|c|} \hline 
            &    Capital Allocation & Risk Sharing  \\ \hline \hline
          Mapping   & ${\bf X} \mapsto {\bf K}\left({\bf X}\right)$ deterministic & ${\bf X} \mapsto {\bf H}({\bf X})$ random\\ \hline 
          Aggregation & \parbox{.33\textwidth}{\centering \vspace{0.1cm} $\sum_{i=1}^nK_i\left({\bf X}\right)=K\left({\bf X}\right)$ \\ Exogenous $K$ by top-down \\ Endogenous $K$ by bottom-up } & $\sum_{i=1}^nH_i({\bf X})=S=\sum_{i=1}^nX_i$ \\          \hline 
        \end{tabular}
        \caption{A comparison between capital allocation and risk sharing strategies}
        \label{tab:compare:CA:RS}
    \end{table}

     In order to induce risk-sharing rules from capital allocation principles, one should first understand the fundamental difference between the two strategies, i.e., allocations are deterministic in capital allocation, but are random in risk-sharing. Our proposed method grounds on the observation that, a risk-sharing rule can be induced if, upon an appropriate randomization, the randomized aggregate capital  coincides with the aggregate risk $S$ in the sharing pool.   Before getting into the details of this approach, we illustrate this idea by the following motivating example. 
        \begin{example}[Marginal VaR, \cite{dhaene:2021}]
        \label{ex:motivate}
            Assume that the marginal distribution functions for all $X_i$, $i=1,\dots,n$, are strictly increasing. 
            For $\boldsymbol{\lambda} = (\lambda_i)_{i=1}^n\in \mathbb{R}^n_+$, define the random variable $S_{\boldsymbol{\lambda}} :=  \sum_{i=1}^n \lambda_i X_i$. For $\theta\in [0,1]$ and $i=1,2,\dots,n$, let
            \begin{equation}
                K_i(\theta,{\bf X}) := \frac{\partial F_{S_{\boldsymbol{\lambda}}}^{-1}(\theta)}{\partial \lambda_i}\bigg|_{\boldsymbol{\lambda}=(1,\dots,1)}.
                \label{eq:motivating_example_allocation}
            \end{equation}

        The vector ${\bf K}(\theta,{\bf X}) = (K_i(\theta,{\bf X}))_{i=1}^n$ is the allocated capital using the Euler principle with the VaR aggregate capital $K(\theta,{\bf X}) = \sum_{i=1}^n K_i(\theta,{\bf X}) = F_S^{-1}(\theta)=\text{VaR}_{\theta}(S)$. In particular, if we evaluate the allocated capitals in \eqref{eq:motivating_example_allocation} at the random parameter $\Theta = F_S(S)$, the resulting random vector ${\bf K}(\Theta,{\bf X})$ coincides with the CMRS rule; that is, $H_i({\bf X})=K_i(\Theta,{\bf X})= \mathbb{E}[X_i|S]$, for $i=1,2,\dots,n$, with $\sum_{i=1}^n K_i(\Theta,{\bf X})=S$ so the sum-preserving property \eqref{eq:sum:preserve} holds.  
        
        \end{example}

    In Example \ref{ex:motivate}, for each $\theta\in[0,1]$, ${\bf K}(\theta,\cdot)$ defines an Euler capital allocation principle with aggregate capital $\text{VaR}_\theta(S)$.  To obtain a risk-sharing rule, the allocations are evaluated at the random parametric level $\Theta = F_S(S)$ so that \eqref{eq:sum:preserve} is fulfilled. This can be understood as a randomization procedure by sampling from a collection of allocation principles  $\{{\bf K}(\theta,\cdot) : \theta \in[0,1]\}$. Motivated by this example, we introduce the following randomization approach to induce risk-sharing rules from capital allocation principles.

\subsection{General Approach}
Let $\mathcal{X}^n$ be a given collection of $\mathbb{R}^n$-valued random vectors with power set $\mathcal{P}(\mathcal{X}^n)$. A capital allocation principle is essentially a $\mathcal{P}(\mathcal{X}^n)$-measurable function from $\mathcal{X}^n$ to $\mathbb{R}^n$. Let $I\subseteq \mathbb{R}$ be a parameter set, assuming that it is an interval. With a slight abuse of notations, denote ${\bf K}$ as a family of allocation principles parameterized in $I$, i.e., ${\bf K} : I\times \mathcal{X}^n \to \mathbb{R}^n$ is a measurable function, and, for each $\theta\in I$, ${\bf K}(\theta,\cdot)=(K_i(\theta,\cdot))_{i=1}^n$ represents a capital allocation principle. A risk-sharing rule induced by the family of allocation principles ${\bf K}$ is defined as follows. 

\begin{definition}
\label{def:induce:risk:sharing}
    Let ${\bf K} : I\times \mathcal{X}^n \to \mathbb{R}^n$ be a family of capital allocation principles. For any ${\bf X} \in \mathcal{X}^n$, let $\Theta({\bf X})$ be a random variable, taking values in $I$, such that
        \begin{equation}
        \label{eq:sum:preserve:induce}
            K(\Theta({\bf X}),{\bf X}) = S,
        \end{equation}
    where  $K(\theta,\cdot)= \sum_{i=1}^n K_i(\theta,\cdot)$, for any $\theta\in I$. The mapping ${\bf K}(\Theta(\cdot),\cdot) : \mathcal{X}^n\to \mathcal{X}^n$ is called a \textit{risk-sharing rule induced by $({\bf K},\Theta)$}, in which the random variable $\Theta(\cdot)$ is called the \textit{parameter sampling rule} for ${\bf K}$. 
\end{definition}

Definition \ref{def:induce:risk:sharing} outlines a legitimate risk-sharing rule ${\bf K}(\Theta(\cdot),\cdot)$, conditioning on that the parameter sampling random variable $\Theta(\cdot)$ exists. Indeed, its condition \eqref{eq:sum:preserve:induce} implies the sum-preserving property \eqref{eq:sum:preserve} required by a risk-sharing rule. The following steps sketch the procedures of inducing such risk-sharing rule, with Figure \ref{fig:random} further illustrating these.
    \begin{enumerate}
        \item (Parametrization). Let $I\subseteq \mathbb{R}$ be the parameter interval. We consider a family of parametrized capital allocation principles ${\bf K}:I\times \mathcal{X}^n\to\mathbb{R}^n$ with, for any $\theta\in I$ and ${\bf X}\in\mathcal{X}^n$, ${\bf K}(\theta,{\bf X})=(K_i(\theta,{\bf X}))_{i=1}^n$ and $K(\theta,{\bf X})=\sum_{i=1}^n K_i(\theta,{\bf X})$. 

        \item (Inversion). {\it Suppose} that, for any ${\bf X}\in\mathcal{X}^n$, $\theta \in I\mapsto K(\theta,{\bf X})$ is surjective on the support $\mathcal{S}_{\bf X}$ of $S$; that is, $\mathcal{S}_{\bf X}\subseteq K(I,{\bf X})$, and it satisfies a regularity condition (e.g., see Proposition \ref{pp:K:surjective} below). Then, for any ${\bf X}\in\mathcal{X}^n$, we can define a measurable right-inverse function $K^{-1+}(\cdot,{\bf X}):\mathcal{S}_{\bf X} \to I$ such that, for any $s\in \mathcal{S}_{\bf X}$, $K(K^{-1+}(s,{\bf X}),{\bf X})=s$. Throughout this paper, we shall refer such property as \textit{measurable right-invertibility}.

        \item (Randomization). Given that, for any ${\bf X}\in\mathcal{X}^n$, a measurable right-inverse of $K(\cdot,{\bf X})$ exists, we define a random variable $\Theta({\bf X}):= K^{-1+}(S,{\bf X})$, which, by definition, satisfies that $K(\Theta({\bf X}),{\bf X})=K(K^{-1+}(S,{\bf X}),{\bf X})=S$. Hence, for any ${\bf X}\in\mathcal{X}^n$,
        \begin{equation}
        \label{eq:risk:sharing:random}
          {\bf H}({\bf X})  = (H_i({\bf X}))_{i=1}^n  := {\bf K}(\Theta({\bf X}),{\bf X}) =(K_i(\Theta({\bf X}),{\bf X}))_{i=1}^n, 
        \end{equation}
   defines a risk-sharing rule, as it satisfies the sum-preserving property \eqref{eq:sum:preserve}.
        
    \end{enumerate}

    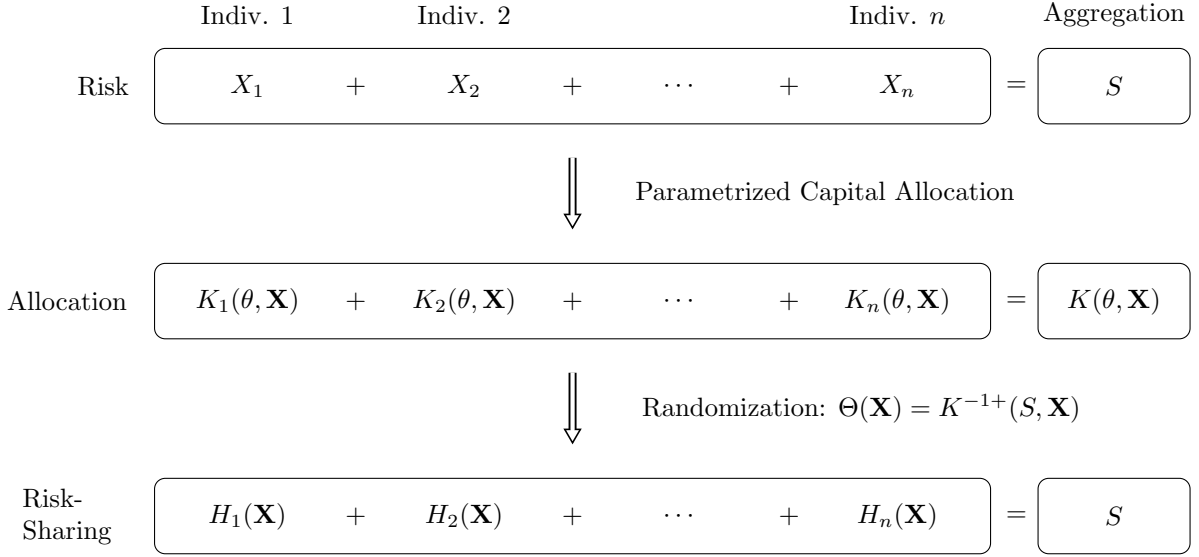
\begin{figure}[!h]
        \centering
        \begin{tikzpicture}[scale=0.95]
            \node at (0,0) {Risk};
            \node at (2,1) {Indiv. 1};
            \node at (5,1) {Indiv. 2};
            \node at (11,1) {Indiv. $n$};
            \node at (14,1) {Aggregation};
            \node (X) [rectangle, rounded corners, minimum width=11cm, minimum height=1cm,text centered, draw=black] at (6.5,0) {  };
            \node at (2,0) {$X_1$};
            \node at (3.5,0) {+};
            \node at (5,0) {$X_2$};
            \node at (6.5,0) {$+$};
            \node at (8,0) {$\cdots$};
            \node at (9.5,0) {$+$};
            \node at (11,0) {$X_n$};

            \node  (S) [startstop] at (14,0) {$S$};
              \draw[vecArrow] (6.5,-1) -- (6.5,-2) ;
              \node at (10,-1.5)  {Parametrized Capital Allocation};
       \node at (-0.5,-3) {Allocation}; 
            \node (K) [rectangle, rounded corners, minimum width=11cm, minimum height=1cm,text centered, draw=black] at (6.5,-3) {  };
           \node at (2,-3) {$K_1(\theta,{\bf X})$};
            \node at (3.5,-3) {+};
            \node at (5,-3) {$K_2(\theta,{\bf X})$};
            \node at (6.5,-3) {$+$};
            \node at (8,-3) {$\cdots$};
            \node at (9.5,-3) {$+$};
            \node at (11,-3) {$K_n(\theta,{\bf X})$}; 
            \node [startstop] at (14,-3) {$K(\theta,{\bf X})$};

              \draw[vecArrow] (6.5,-4) -- (6.5,-5) ;
              \node at (10.5,-4.5)  {Randomization: $\Theta({\bf X})= K^{-1+}(S,{\bf X})$};
    \node[align=left] at (-0.5,-6) {Risk-\\Sharing};

            \node (K) [rectangle, rounded corners, minimum width=11cm, minimum height=1cm,text centered, draw=black] at (6.5,-6) {  };
           \node at (2,-6) {$H_1({\bf X})$};
            \node at (3.5,-6) {+};
            \node at (5,-6) {$H_2({\bf X})$};
            \node at (6.5,-6) {$+$};
            \node at (8,-6) {$\cdots$};
            \node at (9.5,-6) {$+$};
            \node at (11,-6) {$H_n({\bf X})$}; 
            \node [startstop] at (14,-6) {$S$};
            \node at (12.65,0) {=};  
            \node at (12.65,-3) {=};
            \node at (12.65,-6) {=};

        \end{tikzpicture}
        \caption{Illustration of capital-allocation-induced risk sharing rules}
        \label{fig:random}
    \end{figure}

To induce a risk-sharing rule, the underlying family ${\bf K}$ must be sufficiently rich to permit the construction of the parameter sampling rule $\Theta({\bf X})$, for ${\bf X}\in\mathcal{X}^n$, satisfying \eqref{eq:sum:preserve:induce}. This standing assumption requires the aggregate capital function $K(\cdot,{\bf X})$, for ${\bf X}\in\mathcal{X}^n$, to be measurably right-invertible on $\mathcal{S}_{\bf X}$. In Section \ref{sec:risk:aggregation}, we shall present sufficient conditions, which are essentially the surjectivity and the said regularity condition of the aggregate capital function, for this requirement in Proposition \ref{pp:K:surjective}, and showcase examples which satisfy them.

As for injectivity, if the measurably right-invertible function $K(\cdot,{\bf X}):I\to \mathcal{S}_{\bf X}$ is not injective, for ${\bf X}\in\mathcal{X}^n$, its right inverse $K^{-1+}(\cdot,{\bf X}):\mathcal{S}_{\bf X}\to I$ is not unique. Consequently, the parameter sampling random variable $\Theta({\bf X})$, for ${\bf X}\in\mathcal{X}^n$, may not be uniquely determined, inducing multiple risk-sharing rules. In this case, the method of randomization can be characterized by the underlying family ${\bf K}$, along with a specific choice of the measurable right-inverse of the aggregate capital function.

\subsection{Parameter Sampling Rule}

In this randomization approach, the interpretation of the parameter $\theta \in I$ in the family of parametrized capital allocation principles depends on the context. 
For instance, in top-down principles, the aggregate capital $K(\theta,{\bf X})$, for ${\bf X}\in\mathcal{X}^n$, can be exogenously set as a risk measure of the aggregate risk, $\rho_\theta(S)$, evaluated at some parametric level $\theta$. In Example \ref{ex:motivate}, where $\rho_\theta(S)=\text{VaR}_\theta(S)$, for $\theta\in I =[0,1]$, the parameter $\theta$ is the confidence level; 
the induced risk-sharing rule in that example is then the underlying Euler allocation principle evaluated at the random confidence level $\Theta({\bf X}):= F_S(S)$, for ${\bf X}\in\mathcal{X}^n$, such that \eqref{eq:sum:preserve:induce}, and hence \eqref{eq:sum:preserve}, are satisfied.

Whether the capital allocation principles in the family are top-down or bottom-up fundamentally influences how the parameter sampling rule $\Theta(\cdot)$ is determined. In the {\it top-down} principles, just like Example \ref{ex:motivate}, $\Theta(\cdot)$ is fully prescribed by an {\it exogenous} capital aggregation procedure, whereas, in the {\it bottom-up} principles, it arises {\it endogenously} from the capital allocation principles themselves. Section~\ref{sec:risk:aggregation} shall explain this distinction in more details. Because of their differences, the induced risk-sharing rules by the top-down and bottom-up principles shall be discussed in Sections~\ref{sec:top:down} and \ref{sec:bottom-up} respectively.


\subsection{Inheritance of Properties}

When a property is satisfied by all capital allocation principles in the family, the induced risk-sharing rule, by definition, inherits such property in a scenario-wise fashion. That is, for each scenario $\omega\in\Omega$, and thus a realization of ${\bf X}(\omega)$, the realized sharing losses ${\bf H}({\bf X})(\omega)$, which are obtained from the capital allocation principle evaluated at the parameter $\Theta({\bf X})(\omega)$, satisfy the property of the sampled capital allocation principle; that same property still holds when another scenario, and hence a capital allocation principle, is sampled. For example, as it shall be discussed in Section~\ref{sec:optim}, since the capital allocation principles of the optimization approach are Pareto optimal among individual objectives, the risk-sharing rule induced from these principles is Pareto optimal in a scenario-wise sense, instead of a distributional manner. When the induced risk-sharing rule coincides with existing ones in the literature, those properties in a scenario-wise manner offer new perspectives to these risk-sharing rules. For instance, Example \ref{ex:absolute:dev:quantile} shall show that the quantile risk-sharing rule emerges by randomizing the capital allocation principles of the optimization approach with the individual objectives given by the absolute deviation penalty.
 
In particular, such inheritance of properties allows the discussion on the comonotonicity of the induced risk-sharing rules. Given an ${\bf X}\in \mathcal{X}^n$, if the mapping $\theta\mapsto K_i(\theta,{\bf X})$ is either non-increasing or non-decreasing, for all $i=1,\dots,n$, then the induced risk-sharing rule shall be comonotonic. This comonotonic property of the risk-sharing rules induced from the top-down capital allocation principles with distortion risk measures to determine the aggregate capital, and induced from the bottom-up weighted risk allocation principles, shall be connected with the CMRS rule; see Sections \ref{sec:aggregate:distortion} and \ref{sec:weighted:risk}, respectively.

\section{Aggregate Capital Function}\label{sec:risk:aggregation}

In the remaining of this paper, we fix an arbitrary ${\bf X} \in \mathcal{X}^n$ and shall often omit an object's dependence on it whenever there is no ambiguity; we shall write, such as, $K(\theta)$ for $K(\theta,{\bf X})$, $\Theta$ for $\Theta({\bf X})$, and similarly for other quantities.



As discussed in Section \ref{sec:math}, the measurable right-invertibility of the aggregate capital function, $K : I\to \mathbb{R}$, in the parameter $\theta\in I$ on the support $\mathcal{S}_{\bf X}$ of $S$, is crucial to establish the parameter sampling rule associating with the parametric family of capital allocation principles, to induce a legitimate risk-sharing rule. Below, we establish necessary and sufficient conditions for the existence of the parameter sampling rule, which shall be utilized in the next two sections.

    \begin{proposition}
    \label{pp:K:surjective}
    Let ${\bf K}:I\to\mathbb{R}^n$ be the family of the parametrized capital allocation principles in $\theta\in I$, with $K(\theta)=\sum_{i=1}^n K_i(\theta)$. The following statements hold.
    \begin{enumerate}
    \item[(i)] If there exists a random variable $\Theta$ taking values in $I$ such that $K(\Theta)=S$ a.s., then $S \in K(I)$ a.s..
    \item[(ii)] If $K(\cdot)$ is measurable, as well as non-decreasing or continuous, and if $\mathcal{S}_{{\bf X}}\subseteq K(I)$, then $K(\cdot)$ is measurably right-invertible, and there exists a random variable $\Theta$ taking values in $I$ such that $K(\Theta)=S$. 
    \end{enumerate}


    \end{proposition}

\begin{proof}
  See Appendix \ref{sec:app:pf:PP31}. 
\end{proof}

The aggregate capital function $K(\cdot)$, and hence the construction of the parameter sampling rule $\Theta$, depends on the specific capital allocation principles in the family, particularly on whether they follow a top-down or bottom-up principle. As such, the conditions required to ensure the measurable right-invertibility of the aggregate capital function differ between these two principles.

In top-down principles, for each parameter $\theta\in I$, the aggregate capital $K\left(\theta\right)$ is determined {\it exogenously}.
Then, for this given aggregate capital function $K(\cdot)$ in $\theta\in I$, the family of the top-down capital allocation principles ${\bf K}(\cdot)$ satisfies that, for each $\theta\in I$, $\sum_{i=1}^n K_i(\theta) = K(\theta)$. Therefore, the parameter sampling rule $\Theta$ depends solely on the exogenously given aggregate capital function, but not on the capital allocation principles themselves. 
The exogenous aggregate capital function is often obtained by applying a risk measure, at each parametric level $\theta$, to the risk vector ${\bf X}$; that is, for each $\theta\in I$, $K(\theta)=\rho_\theta({\bf X})$, where $\{\rho_\theta\}_{\theta\in I}$ is a parametric family of the risk measures on $\mathcal{X}^n$. In practice, these risk measures usually evaluate the risk vector ${\bf X}$ through its aggregate risk $S$. Hence, when the capital allocation principles in a family to induce a risk-sharing rule are top-down, the following problem shall be addressed.
\begin{problem}
\label{p1}
 Find a parametric family of risk measures $\{\rho_\theta\}_{\theta\in I}$ such that the aggregate capital function $K:I\to \mathbb{R}$, defined by $K(\theta)=\rho_\theta(S)$, is measurably right-invertible on $\mathcal{S}_{\bf X}$, i.e., 
 there exists a measurable right-inverse function, $K^{-1+}:\mathcal{S}_{{\bf X}} \to I$.       
\end{problem}

     
       
       For bottom-up principles, the aggregate capital function is not given by an exogenous function. Instead, the aggregate capital is determined {\it endogenously}, summing up from the allocated capitals.  In this case, $K(\cdot)$ generally depends on the entire random vector ${\bf X}$ and is given by $K(\theta) = \rho_\theta({\bf X})$, where the family $(\rho_\theta)_{\theta\in I}$ is endogenously determined by the allocation mechanism. For example, if $K(\theta)$ is obtained through an optimization problem, then the family $(\rho_\theta)_{\theta\in I}$ corresponds to the solutions to this family of problems. Consequently, the parameter sampling rule $\Theta$ also depends on the underlying allocation mechanism itself.
       
       Figure~\ref{fig:td:bu} illustrates the difference between the top-down and bottom-up principles in determining $\Theta$ and the resulting risk-sharing rules. Due to their differences, 
       we shall delegate the discussion of randomizing a bottom-up principle in Section \ref{sec:bottom-up}. In the remaining of this section as well as the next Section \ref{sec:top:down}, we focus on Problem \ref{p1} for the top-down principles.

\begin{figure}
    \centering
    \begin{tikzpicture}[
    >=Stealth,
    node distance=1.5cm and 3cm,
    block/.style={
        rectangle, 
        draw, 
        rounded corners, 
        align=center, 
        minimum width=4cm, 
        minimum height=1cm
    }
]

\node[block, fill=blue!10,minimum width=5.5cm] (exanteTD) {
    \textbf{\textit{Ex-ante} stage} \\
  Specify:  \\
  \makebox[6cm][l]{1. Aggregate capital $K(\theta)$}\\
    \makebox[6cm][l]{2. Sampling rule $\Theta$ based on $K$}\\
    \makebox[6cm][l]{3. Allocation rules $\mathbf{K}(\theta)$ s.t.}\\
    \makebox[6cm][l]{\hspace*{1.5em}$\sum_{i=1}^n K_i(\theta)=K(\theta)$}
};

\node[block, fill=green!10, right=3.0cm of exanteTD, minimum width = 7.1cm] (expostTD) {
    \textbf{\textit{Ex-post} stage} \\
    Find $\Theta(\omega)$ s.t. $K(\Theta(\omega)) = S(\omega)$ \\
    $\Downarrow$ \\
     Determine loss sharing: ${\bf K}(\Theta(\omega))$
};

\draw[->] (exanteTD) -- node[above]{${\bf X}(\omega)$ realized} (expostTD);

\node[above=0.3cm of exanteTD] (BP) {\large \textbf{Before Peril}};
\node[right=6.5cm of BP] {\large \textbf{After Peril}};
\node[below right=0.2cm and 0cm of exanteTD.south east] {(a) Top-down principle};

\node[block, fill=blue!10, minimum width=5.5cm,  below=2cm of exanteTD] (exanteBU) {
    \textbf{\textit{Ex-ante} stage} \\
   \makebox[6cm][l]{1. Specify allocation rules $\mathbf{K}(\theta)$}\\
  \makebox[6cm][l]{2. Determine aggregate capital} \\
   \makebox[6cm][c]{s.t.~$K(\theta) = \sum_{i=1}^n K_i(\theta)$}\\
    \makebox[6cm][l]{3.  Specify  $\Theta$ based on $K$}
     
};

\node[block, fill=green!10, right=3cm of exanteBU, minimum width = 7cm] (expostBU) {
    \textbf{\textit{Ex-post} stage} \\
    Find $\Theta(\omega)$ s.t. $K(\Theta(\omega))=S(\omega) $ \\
    $\Downarrow$ \\
    Determine loss sharing: ${\bf K}(\Theta(\omega))$
};

\draw[->] (exanteBU) -- node[above]{${\bf X}(\omega)$ realized} (expostBU);

\node[below right=0.2cm and 0cm of exanteBU.south east] {(b) Bottom-up principle};

\end{tikzpicture}
 
    \caption{Implementation of induced risk-sharing rules based on top-down and bottom-up principles}
    \label{fig:td:bu}
\end{figure}
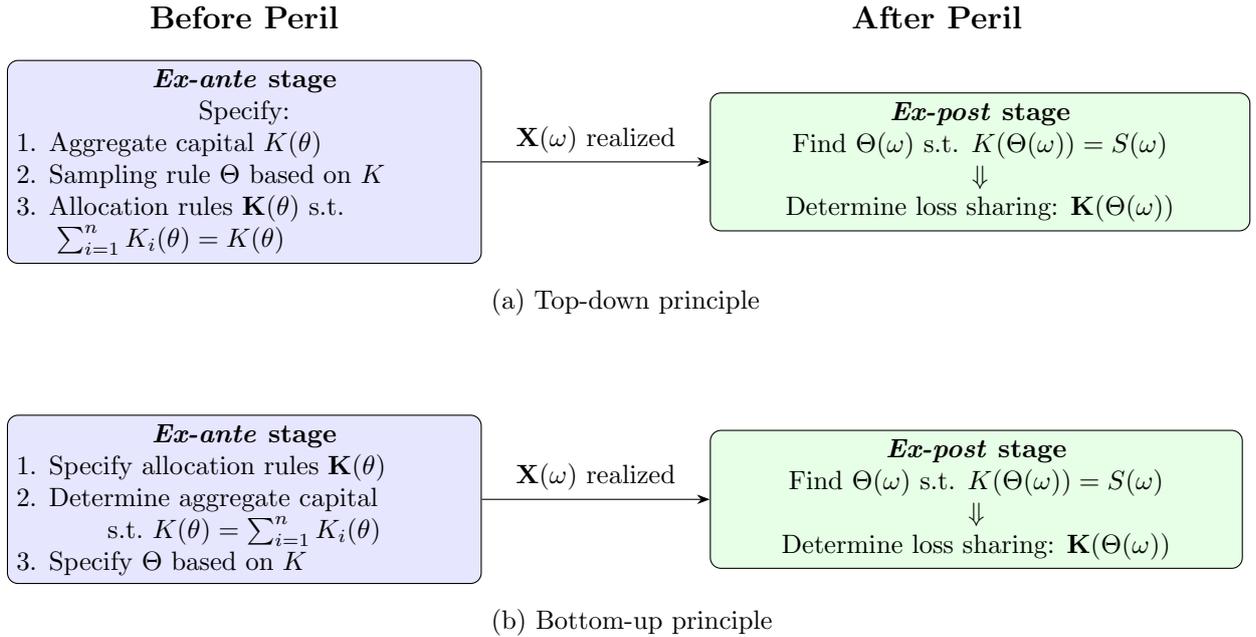

       Apparently, not all families of the risk measures on the aggregate risk $S$ are surjective in $\theta\in I$. For instance, if the risk measures in a family $(\rho_\theta)_{\theta\in I}$ have a non-negative risk loading, i.e., for any $\theta\in I$, $\rho_\theta(S) \geq \mathbb{E}[S]$, and thus $K(\theta)\not\in [F_S^{-1+}(0),\mathbb{E}[S])$. In this case, $\mathcal{S}_{\bf X} \not\subseteq K(I)$, unless $\mathcal{S}_{\bf X}$ is a singleton (from degenerate risks).
       
       In contrast, without prescribing any properties on the underlying risk measures, one trivial example for the family that satisfies Problem \ref{p1} can be constructed as follows.  
        \begin{example}
        \label{ex:trivial}
            Let $\rho_\theta(S) = f(\theta) \rho(S) $, for $\theta\in I$, where $\rho$ is a  risk measure with $\rho(S) \in\mathbb{R}\backslash \{0\}$, and $f : I \to \mathbb{R}$ is bijective. Then, by taking $\Theta:= f^{-1}(S/\rho(S))$, one immediately has $\rho_\Theta(S)=S$. 
        \end{example}
    \noindent
    The construction in Example \ref{ex:trivial} is simple. However, the risk measures do not satisfy some basic desirable properties, for example, translational invariant. In addition, the parameter $\theta\in I$ is separated from the distribution of $S$ in the risk measures, which is restrictive and lacks interpretability. In Section \ref{sec:aggregate:distortion}, we shall focus on distortion risk measures that naturally satisfy the usual desirable properties, such as translational invariant and monotonicity. In addition, these risk measures exhibit the measurable right-invertibility in $\theta\in I$, under tractable conditions on the underlying distortion functions, which shall be shown using Proposition \ref{pp:K:surjective}. Below, we provide an example with a class of quantile-based risk measures that also fulfill the measurably right-invertible requirement.

    \begin{example}
        The VaR, $\theta\in[0,1]\mapsto\rho_\theta(S)=\text{VaR}_\theta(S)$, is measurable, non-decreasing and surjective on $\mathcal{S}_{{\bf X}}$, i.e., it satisfies the sufficient conditions in Proposition \ref{pp:K:surjective}. More generally, let $w : \mathbb{R} \to (0,\infty)$ be a positive function such that $\int_{\mathcal{S}_{{\bf X}}} w(s) dF_S(s) < \infty$. Then, the weighted quantile, defined by 
            \begin{equation*}
             \rho^w_\theta(S)   := \inf\left\{ s : F^w(s):= \frac{\int_{-\infty}^s w(t)dF_S(t) }{\int_{\mathcal{S}_{{\bf X}}} w(t)dF_S(t) } \geq \theta   \right\}, \ \theta\in[0,1],
            \end{equation*}
        is measurable, non-decreasing, and surjective onto $\mathcal{S}_{{\bf X}}$. Indeed, $\rho^w_\theta(S) = \text{VaR}_\theta(S^w)$, where $S^w$ is a random variable with distribution function $F^w(\cdot)$ and the same support $\mathcal{S}_{\bf X}$, which implies the properties.
        
    \end{example}

\section{Randomization of Top-Down Principles}
\label{sec:top:down}
 In this section, we shall consider different top-down principles of capital allocation and induce their associated risk-sharing rules by the randomization procedures. The randomized optimization principle and the randomized Euler principle are introduced in Sections \ref{sec:optim} and \ref{sec:euler}, respectively. In Section \ref{sec:distortion}, we discuss the case when the aggregate capital function is defined by a family of parametrized distortion risk measures.

\subsection{Optimization Principle}
\label{sec:optim}
The optimization principle of capital allocations was introduced in \cite{dhaene}. The method seeks an allocation by minimizing the sum of expected penalties (or costs) on the deviations between individual risks and the allocated capitals with user-defined parameters and penalty functions. Specifically, given an aggregate capital $K$, the allocated capitals are determined as the minimizer of the following optimization problem:
    \begin{equation}
    \label{eq:min}
        \inf_{\substack{K_1,\dots,K_n\\ \sum_{i=1}^nK_i=K} } \sum_{i=1}^n\beta_i \mathbb{E}^{\mathbb{Q}_i}\left[ d\left( \frac{X_i-K_i}{\beta_i} \right) \right],
    \end{equation}
where $d\left(\cdot\right)$ is a non-negative convex penalty function, $(\beta_i)_{i=1}^n$ are non-negative constants with $\sum_{i=1}^n\beta_i=1$ that represent the relative exposures of each LOB, and $(\mathbb{Q}_i)_{i=1}^n$ are probability measures which reflect the preferences of the LOBs. As the optimization principle is top-down, capital aggregation and allocation are completely decoupled, with each step employing a method that neither relies on nor influences the other one.

If the measures  $(\mathbb{Q}_i)_{i=1}^n$ are equivalent to $\mathbb{P}$, they can appear in the form of Radon-Nikodym derivatives as follows: for $i=1,2,\dots,n$,
    \begin{equation*}
        \mathbb{E}^{\mathbb{Q}_i}\left[ d\left( \frac{X_i-K_i}{\beta_i} \right) \right] = \mathbb{E}\left[ \zeta_i d\left( \frac{X_i-K_i}{\beta_i} \right) \right],
    \end{equation*}
where $\zeta_i$ is a non-negative random variable satisfying $\mathbb{E}[\zeta_i]=1$; its common choice includes, $\zeta_i=h_i(X_i)$ and $\zeta_i=h_i(S)$, where $h_i(\cdot)$ is a non-negative function. In \cite{dhaene}, Problem \eqref{eq:min} was solved for the cases, when $d(z)=z^2$ and $d(z)=|z|$, for $z\in\mathbb{R}$. Herein, we shall deduce  risk-sharing rules based on the capital allocation principles in \eqref{eq:min} for these two choices of the penalty function.

To do so, we consider the following parametrized version of \eqref{eq:min}: for $\theta\in I$, solve 
  \begin{equation}
    \label{eq:min:para}
      {\bf K}(\theta) :=  \mathop{\arg\min}_{\substack{K_1,\dots,K_n\\ \sum_{i=1}^nK_i=K(\theta)} } \sum_{i=1}^n\beta_i(\theta) \mathbb{E}^{\mathbb{Q}_i^\theta}\left[ d\left( \frac{X_i-K_i}{\beta_i(\theta)} \right) \right],
    \end{equation}
where $K: I\to \mathbb{R}$ is exogenously given and is measurably right-invertible. In this case, there exists a random variable $\Theta$ taking values in $I$ such that $K(\Theta)=S$, and a risk-sharing rule is induced by $({\bf K},\Theta)$.
The rule is obtained by first solving the optimization problem \eqref{eq:min:para} for each $\theta\in I$, followed by evaluating the capitals ${\bf K}(\cdot)$ at the random parametric level $\Theta$.

The induced rule has the following interpretation. When rules are determined \textit{ex-ante}, for each possible realization of the aggregate risk, a deterministic allocation rule ${\bf K}(\theta)=(K_i(\theta))_{i=1}^n$ is constructed so that the sum of expected deviations from the individual risks under that realization is minimized. When $\omega\in\Omega$, and thus $S\left(\omega\right)=s$, are realized \textit{ex-post}, the shared losses are computed using the parametric level $\Theta\left(\omega\right)$ such that the exogenous aggregate allocation $K\left(\Theta\left(\omega\right)\right)$ agrees with $S\left(\omega\right)=s$ in this realization.

In the parametrization \eqref{eq:min:para}, the capital aggregation and allocation are connected by the parameter $\theta\in I$. Consequently, the induced risk-sharing rule generally depends on both the underlying capital aggregation and allocation methods. But, in particular, if the coefficients $\beta_i$ and the preference measures $\mathbb{Q}_i$ are chosen to be independent of $\theta$, the family ${\bf K}(\theta)$ depends on $\theta$ only through $K(\theta)$, i.e., we can represent the underlying family of allocations as ${\bf K}(\theta)=g(K(\theta))$, for some function $g(\cdot)$ which does not depend on the choice of the aggregate capital function. As a result, the induced risk-sharing rule, ${\bf H}({\bf X})={\bf K}(\Theta)=g(K(\Theta))=g(S)$, is independent of the capital aggregation method, as long as the aggregate capital function is measurably right-invertible; for instance, see Example \ref{ex:square:separate} below.


The optimization principle yields a Pareto-optimal capital allocation.
\begin{definition}
\label{def:PO:allocation}
Let $\mathcal{A}\subseteq\mathbb{R}^n$. A vector ${\bf K}^*=(K_i^*)_{i=1}^n\in\mathcal{A}$ is said to be Pareto optimal in $\mathcal{A}$ with respect to the functions $(\varphi_i(\cdot))_{i=1}^n$, if there does not exist ${\bf K}=(K_i)_{i=1}^n \in\mathcal{A}$ such that $\varphi_i({\bf K}) \le \varphi_i({\bf K}^*)$, for all $i=1,\dots,n$, with at least one inequality being strict.
\end{definition}
It is rather straightforward to prove that, if there exists ${\bf K}^*\in\mathcal{A}$ that minimizes $\inf_{{\bf K}\in\mathcal{A}} \sum_{i=1}^n \varphi_i({\bf K})$, then ${\bf K}^*$ is Pareto optimal. Hence, for each $\theta\in I$, the capital allocation ${\bf K}(\theta)$, which solves the optimization principle \eqref{eq:min:para}, is Pareto optimal in the feasible set $\mathcal{A}=\mathcal{K}_{K(\theta)}=\{{\bf K}\in\mathbb{R}^n : \sum_{i=1}^n K_i = K(\theta)\}$ with respect to the cost functions:
\begin{equation}
\label{eq:varphi:i:theta}
\varphi_i^\theta({\bf K}) = \beta_i(\theta) \mathbb{E}^{\mathbb{Q}_i^\theta}\left[d\left(\frac{X_i-K_i}{\beta_i(\theta)} \right)\right],\;i=1,\dots,n.
\end{equation}

A natural question is then whether the induced risk-sharing rule ${\bf K}(\Theta)$ is Pareto optimal, where ${\bf K}(\theta)$, for $\theta\in I$, solves the optimization principle \eqref{eq:min:para}, and $\Theta=K^{-1+}(S)$ is the parameter sampling rule such that $K(\Theta)=S$. Fix an $\omega\in\Omega$; a parameter $\Theta(\omega)=K^{-1+}(S)(\omega)$ is then sampled, and hence the risk-sharing rule is realized as ${\bf K}(\Theta(\omega))$, which is a capital allocation in the family, from the optimization principle \eqref{eq:min:para}, entailing how the aggregate loss $K(\Theta)(\omega)=S(\omega)$ should be shared ex-post. Therefore, with the fixed $\omega\in\Omega$, ${\bf K}(\Theta(\omega))$ is Pareto optimal in the feasible set $\mathcal{K}_{K(\Theta)(\omega)}=\{{\bf K}\in\mathbb{R}^n : \sum_{i=1}^n K_i = K(\Theta)(\omega)=S(\omega)\}$ with respect to the cost functions:
\begin{equation}
\label{eq:varphi:i:theta_omega}
\varphi_i^{\Theta(\omega)}({\bf K}) = \beta_i(\Theta(\omega)) \mathbb{E}^{\mathbb{Q}_i^{\Theta(\omega)}}\left[d\left(\frac{X_i-K_i}{\beta_i(\Theta(\omega))} \right)\right],\;i=1,\dots,n.
\end{equation}
This observation motivates the introduction of Pareto optimality in a scenario-wise manner, which is defined as follows.

\begin{definition}
\label{def:pareto:optimal:pointwise}
Let $\mathcal{Y}$ be the set of risk-sharing rules which are random vectors ${\bf Y}=(Y_i)_{i=1}^n:\Omega\rightarrow\mathbb{R}^n$ such that $\sum_{i=1}^{n}Y_i=S$. Let $\varphi_i:\Omega\times\mathbb{R}^n\rightarrow\mathbb{R}$, for $i=1,\dots,n$. A risk-sharing rule ${\bf Y}^*=(Y^*_i)_{i=1}^n\in\mathcal{Y}$ is said to be scenario-wise Pareto optimal in $\mathcal{Y}$ with respect to the functions $(\varphi_i(\cdot,\cdot))_{i=1}^n$, if, for a.a. $\omega\in\Omega$, ${\bf Y}^*(\omega)$ is Pareto optimal in $\mathcal{Y}(\omega)=\left\{{\bf Y}(\omega)\in\mathbb{R}^n:{\bf Y} \in \mathcal{Y}\right\}$ with respect to $(\varphi_i(\omega,\cdot))_{i=1}^n$, in the sense of Definition \ref{def:PO:allocation}.
\end{definition}
\noindent
Based on this definition, the induced risk-sharing rule ${\bf K}(\Theta)$ is indeed scenario-wise Pareto optimal among the risk-sharing rules with respect to the cost functions given in \eqref{eq:varphi:i:theta_omega}.

The concept of scenario-wise Pareto optimality is different from the classical notion of Pareto optimality, which is recalled as follows.
\begin{definition}
\label{def:PO:RS}
Let $\mathcal{Y}$ be the set of risk-sharing rules which are random vectors ${\bf Y}=(Y_i)_{i=1}^n:\Omega\rightarrow\mathbb{R}^n$ such that $\sum_{i=1}^{n}Y_i=S$. Let $\rho_i:\mathcal{Y}\rightarrow\mathbb{R}$, for $i=1,\dots,n$. A risk-sharing rule ${\bf Y}^*=(Y^*_i)_{i=1}^n\in\mathcal{Y}$ is said to be Pareto optimal in $\mathcal{Y}$ with respect to the functions $(\rho_i(\cdot))_{i=1}^n$, if there does not exist ${\bf Y}=(Y_i)_{i=1}^n\in\mathcal{Y}$ such that $\rho_i({\bf Y}) \le \rho_i({\bf Y}^*)$, for all $i=1,\dots,n$, with at least one inequality being strict.
\end{definition}%
\noindent
Again, it is clear that, if there exists ${\bf Y}^*\in\mathcal{Y}$ that minimizes $\inf_{{\bf Y}\in\mathcal{Y}} \sum_{i=1}^n \rho_i({\bf Y})$, then ${\bf Y}^*$ is Pareto optimal.

The key characteristic which distinguishes these two notions of Pareto optimality lies in their evaluation functions. In Definition \ref{def:pareto:optimal:pointwise} for the scenario-wise Pareto optimality, the functions $(\varphi_i(\cdot,\cdot))_{i=1}^n$ evaluate a risk-sharing rule ${\bf Y} \in \mathcal{Y}$ by, first sampling $\omega\in\Omega$, and then evaluating its realization ${\bf Y}(\omega)\in\mathbb{R}^n$; that is, the evaluation is for each outcome. In Definition \ref{def:PO:RS} for the classical notion of Pareto optimality, the functions $(\rho_i(\cdot))_{i=1}^n$ evaluate a risk-sharing rule ${\bf Y} \in \mathcal{Y}$ without sampling any outcomes; that is, the evaluation directly assigns at the aggregate level.


The induced risk-sharing rule ${\bf K}(\Theta)$ is also Pareto optimal in the sense of Definition \ref{def:PO:RS} with respect to the functions: for any ${\bf Y}\in\mathcal{Y}$,
\begin{equation*}
\rho_i({\bf Y})=\int_\Omega  \varphi_i(\omega,Y_i(\omega))\, \mu(d\omega),\;i=1,\dots,n,
\end{equation*}
where $\mu$ is a finite measure defined on $(\Omega,\mathcal{F})$ and is absolutely continuous with respect to $\mathbb{P}$, and assuming that $\varphi_i(\cdot,Y_i(\cdot))$ is $\mu$-integrable. Indeed, for any ${\bf Y}\in\mathcal{Y}$,
\begin{equation*}
\sum_{i=1}^n \rho_i({\bf Y})= \int_\Omega \sum_{i=1}^n  \varphi_i(\omega,Y_i(\omega))\,\mu(d\omega)\geq \int_\Omega \sum_{i=1}^n \varphi_i(\omega,K_i(\Theta(\omega)))\,\mu(d\omega) = \sum_{i=1}^n\rho_i(K_i(\Theta)),
\end{equation*}
where the inequality is due to \eqref{eq:min:para}. Therefore, ${\bf Y}^*={\bf K}(\Theta)$ minimizes $\inf_{{\bf Y}\in\mathcal{Y}} \sum_{i=1}^n \rho_i({\bf Y})$, and thus ${\bf Y}^*={\bf K}(\Theta)$ is Pareto optimal.

\subsubsection{Squared penalty}

When $d(z)=|z|^2$, for $z\in\mathbb{R}$, the optimization problem \eqref{eq:min:para} can be written  as follows: for $\theta\in I$,
    \begin{equation}
    \label{eq:min:square}
        {\bf K}(\theta) = \mathop{\arg\min}_{\substack{K_1,\dots,K_n \\ \sum_{i=1}^nK_i = K(\theta)}} \sum_{i=1}^n \frac{1}{\beta_i(\theta)} \mathbb{E}^{\mathbb{Q}_i^\theta}[ (X_i-K_i)^2 ].
    \end{equation}
Here, we assumed that $X_i\in L^2(\Omega,\mathcal{F},\mathbb{Q}^\theta_i)$, for $i=1,\dots,n$ and $\theta \in I$, and $K:I \to \mathcal{S}_{\bf X}$ is measurably right-invertible; there exists a random variable $\Theta$ taking values in $I$ such that $K(\Theta)=S$.

To write down explicitly the $({\bf K},\Theta)$-induced risk-sharing rule, we recall the solution of the optimization problem \eqref{eq:min:square}, for each $\theta\in I$; by \cite{dhaene}, we have
    \begin{equation}
    \label{eq:capital:square}
        K_i(\theta)  = \mathbb{E}^{\mathbb{Q}_i^\theta}[X_i] + \beta_i(\theta) \left( K(\theta) - \sum_{j=1}^n  \mathbb{E}^{\mathbb{Q}_j^\theta}[X_j] \right), \ i=1,2,\dots,n.
    \end{equation}
Randomizing \eqref{eq:capital:square} by the random parameter $\Theta$ yields the following risk-sharing rule:
    \begin{equation}
    \label{eq:risk:square}
        H_i({\bf X}) = K_i(\Theta) = \mathbb{E}^{\mathbb{Q}_i^\theta}[X_i]  \bigg|_{\theta=\Theta} + \beta_i(\Theta)   \left( S - \sum_{j=1}^n  \mathbb{E}^{\mathbb{Q}_j^\theta}[X_j] \bigg|_{\theta=\Theta}  \right),\ i=1,2,\dots,n.
    \end{equation}

\begin{remark}
    The induced risk-sharing rule \eqref{eq:risk:square} is different from the solution of the following problem:
    \begin{equation*}
        \inf_{\substack{H_1({\bf X}),\dots,H_n({\bf X})\\ \sum_{i=1}^n H_i({\bf X)}=S } } \sum_{i=1}^n \frac{1}{\beta_i(\theta)}\mathbb{E}^{\mathbb{Q}_i^\theta}\left[\left(X_i-H_i({\bf X})\right)^2\right] \bigg|_{\theta=\Theta},
    \end{equation*}
which directly searches for a risk vector that sums to the aggregate risk $S$ and minimizes the sum of squared penalties; obviously, a trivial minimizer of this problem is given by $H_i({\bf X})=X_i$, for $i=1,2,\dots,n$. In contrast, \eqref{eq:risk:square} is obtained by randomizing the deterministic solution \eqref{eq:min:square}, or equivalently \eqref{eq:capital:square}, and the resulting shared loss is a deterministic function of $S(\omega)$.
\end{remark}    

\begin{example}
\label{ex:square:separate}
Suppose that the relative exposures $(\beta_i)_{i=1}^n$ and the preference measures $(\mathbb{Q}_i)_{i=1}^n$ are independent of $\theta$. Then, the risk-sharing rule in \eqref{eq:risk:square} is reduced to
\begin{equation}
    \label{eq:opt:square}
        H_i({\bf X})  = \mathbb{E}^{\mathbb{Q}_i}[X_i]   + \beta_i   \left( S - \sum_{j=1}^n  \mathbb{E}^{\mathbb{Q}_j}[X_j]   \right),  \ i=1,2,\dots,n,
    \end{equation}
which is an example of the quota-share risk-sharing rule. That is, the quota-share rule in the form \eqref{eq:opt:square} is the optimal way to share the losses such that the sum of squared penalties is minimized when $S$ is realized. The rule \eqref{eq:opt:square} depends on the marginal distribution of ${\bf X}$ and on the dependence structure only through the aggregate risk $S$, rather than on features such as pairwise covariances. 
The reason is clear; the optimization problem \eqref{eq:min:square} only considers the sum of individual squared errors. Effectively, due to the separation of the capital aggregation and allocation procedures, and the assumption that the relative exposures and preference measures are independent of the parametrization, the risk-sharing rule is induced by simply replacing the aggregate capital $K(\theta)$ in the allocation rule by $S$. Consequently, the underlying method of capital aggregation no longer plays a role in the induced rule as long as it has a measurable right-inverse. Inherited from the allocation principle that the capitals are non-decreasing functions of $K(\theta)$, the resulting risk-sharing rule \eqref{eq:opt:square} is comonotonic, as the risks are non-decreasing functions of $S$. In addition, the sharing rule is scenario-wise Pareto-optimal as defined in Definition \ref{def:pareto:optimal:pointwise}. 

However, suppose that either a relative exposure, or a preference measure, depends on $\theta\in I$. After randomization, the induced risk-sharing rule \eqref{eq:risk:square} is no longer a quota-share rule, as it depends on the method of capital aggregation through the parameter sampling rule $\Theta$, which is the measurable right-inverse of the aggregate capital function on the aggregate risk $S$, appearing in the relative exposure or the preference measure. Also, because of this, the resulting risk-sharing rule \eqref{eq:risk:square} is not necessarily comonotonic, but is still scenario-wise Pareto-optimal as defined in Definition \ref{def:pareto:optimal:pointwise}. The following example further illustrates this.


\end{example}

\begin{example}\label{eg:ellip}
  Suppose that ${\bf X}$ follows an elliptical distribution, denoted by ${\bf X}\sim E_n(\boldsymbol{\mu},\boldsymbol{\Sigma},g)$, with a mean vector $\boldsymbol{\mu}=(\mu_i)_{i=1}^n$, a positive definite covariance matrix $\boldsymbol{\Sigma}$, and a density generator $g:\mathbb{R}_+:=[0,\infty)\to \mathbb{R}_+$; i.e., the joint density of ${\bf X}$ is given by: for ${\bf x}\in\mathbb{R}^n$,
    \begin{equation*}
        f_{{\bf X}}({\bf x}) = c_n |{\bf \Sigma}|^{-\frac{1}{2}}g\left(\frac{1}{2}({\bf x-\boldsymbol{\mu}})^\top \boldsymbol{\Sigma}^{-1} ({\bf x-\boldsymbol{\mu}}) \right),
    \end{equation*}
where $c_n>0$ is a normalizing constant which depends on $n$, but the generator $g$ is independent of $n$. 
  Define the cumulative generator $G$ and its complement $\bar{G}$ by: for $z\geq 0$,
    \begin{equation*}
        G(z) := \int_0^z g(t)\,dt,  \quad \text{and} \quad \bar{G}(z) = G(\infty)- G(z).  
    \end{equation*}
  

  Consider the parametrization $\theta\in I=[0,1]$, with the aggregate capital function $K : I^\circ = (0,1) \to \mathbb{R}$ given by $K(\theta)=\text{VaR}_\theta(S)$, $K(0) = F_S^{-1+}(0)$, and $K(1) = F_S^{-1}(1)$. The function $K(\cdot)$ is measurable, non-decreasing, and surjective onto $\mathcal{S}_{\bf X}$, and hence admits a measurable right-inverse by Proposition \ref{pp:K:surjective}. For $\theta\in[0,1)$, the preference measures $(\mathbb{Q}_i^\theta)_{i=1}^n$ are characterized by the Radon-Nikodym derivatives $(h_{i,\theta})_{i=1}^n$ with respect to $\mathbb{P}$, such that $\mathbb{E}^{\mathbb{Q}_i^\theta}[X_i] = \mathbb{E}[ X_i h_{i,\theta}(X_i)]$, for $i=1,2,\dots,n$, where
        \begin{equation*}
            h_{i,\theta}(X_i) = \frac{\mathbbm{1}_{\{X_i>F_{X_i}^{-1}(\theta)\} }}{ 1-\theta}.
        \end{equation*}
    By \cite{landsman2003tail}, for $\theta\in[0,1)$,
        \begin{equation*}
            \mathbb{E}^{\mathbb{Q}_i^\theta}[X_i] = \text{TVaR}_{\theta}(X_i) = \mu_i + \frac{\sigma_i \bar{G}\left(\frac{1}{2}z_{i,\theta}^2\right)}{1-\theta},
        \end{equation*}
    where $z_{i,\theta}:=F_{Z_i}^{-1}(\theta)$ and $(Z_1,\dots,Z_n):= {\bf \Sigma}^{-\frac{1}{2}}({\bf X}-\boldsymbol{\mu})$, whereas $\mathbb{E}^{\mathbb{Q}_i^\theta}[X_i] = \text{TVaR}_{\theta}(X_i)=\text{ess}\sup(X_i)$, for $\theta=1$.
  
    By the fact that $\text{VaR}_\Theta(S)=K(\Theta)=S$, we have $\Theta=F_S(S)$, and 
    the risk-sharing rule \eqref{eq:risk:square} reads, for $i=1,\dots,n$, 
        \begin{equation*}
            H_i({\bf X}) = \mu_i + \frac{\sigma_i \bar{G}\left(\frac{1}{2}z_{i,F_S(S)}^2\right)}{1-F_S(S)}+\beta_i(\Theta) \left[S - \left( \mu_S + \sum_{j=1}^n \frac{\sigma_j \bar{G}\left(\frac{1}{2}z_{j,F_S(S)}^2\right)}{1-F_S(S)} \right)\right],
        \end{equation*}
    when $S<\text{ess}\sup(S)$; when $S=\text{ess}\sup(S)$, $H_i({\bf X})=\text{ess}\sup(X_i)$. Due to the randomization of the relative exposures $(\beta_i)_{i=1}^n$ and the preference measures $(\mathbb{Q}_i)_{i=1}^n$, the resulting risk-sharing rule is no longer affine in $S$. In addition, it depends on the aggregate capital function $K(\cdot)=F_S^{-1}(\cdot)$ via its measurable right-inverse $K^{-1+}(\cdot)=F_S(\cdot)$.
\end{example}

\subsubsection{Absolute penalty}
When $d(z)=|z|$, for $z\in\mathbb{R}$, the optimization problem \eqref{eq:min:para} can be written as: for $\theta\in I$,
    \begin{equation}
    \label{eq:min:abs}
        {\bf K}(\theta) =\mathop{\arg\min}_{\substack{K_1,\dots,K_n \\ \sum_{i=1}^nK_i=K(\theta)}} \sum_{i=1}^n \mathbb{E}^{\mathbb{Q}_i^\theta}[|X_i-K_i|]. 
    \end{equation}
Here, we assumed that $X_i\in L^1(\Omega,\mathcal{F},\mathbb{Q}^\theta_i)$, for $i=1,\dots,n$ and $\theta \in I$, and $K:I \to \mathcal{S}_{\bf X}$ is measurably right-invertible; there exists a random variable $\Theta$ taking values in $I$ such that $K(\Theta)=S$.



To write down the solution of \eqref{eq:min:abs}, we  define the generalized quantiles as follows. For $i=1,2,\dots,n$ and $\theta\in I$, let $F^{\mathbb{Q}_i^\theta}_i(x) :=\mathbb{Q}_i^\theta(X_i\leq x)$, for $x\in\mathbb{R}$, be the distribution function of $X_i$ under the measure $\mathbb{Q}_i^\theta$, and $(F^{\mathbb{Q}^\theta_i}_i)^{-1}(\cdot)$ be the corresponding left-inverse quantile function. Define a random variable $\bar{S}^c_\theta$, for $\theta\in I$, by
 \begin{equation*}
        \bar{S}^c_\theta = \sum_{i=1}^n \left(F^{\mathbb{Q}_i^\theta}_i\right)^{-1}(U),
    \end{equation*}
where $U$ is a uniform random variable over the interval $[0,1]$; 
hence, $\bar{S}^c_\theta$ is a sum of comonotonic random variables, where each random variable $\left(F^{\mathbb{Q}_i^\theta}_i\right)^{-1}(U)$ has a distribution function given by $F^{\mathbb{Q}_i^\theta}_i(\cdot)$. For $\alpha\in[0,1]$, we define the $\alpha$-mixed inverse of $F_{\bar{S}^c_\theta}$, for $\theta\in I$, by: for $u\in\left[0,1\right]$,
 \begin{equation*}
        F^{-1(\alpha)}_{\bar{S}^c_\theta }(u) := \alpha F^{-1}_{\bar{S}^c_\theta}(u) +  (1-\alpha)F^{-1+}_{\bar{S}^c_\theta}(u);
  \end{equation*}
then, for any $s \in [F^{-1+}_{\bar{S}^c_\theta}(0) ,F^{-1}_{\bar{S}^c_\theta}(1)]$, there exists $\alpha_s \in [0,1]$ such that  $ F^{-1(\alpha_s)}_{\bar{S}^c_\theta }(s)=s$.

By \cite{dhaene}, the solution of \eqref{eq:min:abs} can be expressed as, for $\theta\in I$,
    \begin{equation}
    \label{eq:min:sol}
        K_i(\theta) =  \left(F^{\mathbb{Q}_i^\theta }_i\right)^{-1(\alpha_{K(\theta)})}(F_{\bar{S}^c_\theta}(K(\theta))), \ i=1,2,\dots,n,
    \end{equation}
where $\alpha_{K(\theta)}\in[0,1]$ such that $ F^{-1(\alpha_{K(\theta)})}_{\bar{S}^c_\theta }(F_{\bar{S}^c_\theta}(K(\theta)))=K(\theta)$, provided that $K(\theta) \in [F^{-1+}_{\bar{S}^c_\theta}(0) ,F^{-1}_{\bar{S}^c_\theta}(1)]$. 
Based on the capital allocation \eqref{eq:min:sol}, it is straightforward to induce the risk-sharing rule via substituting the parameter $\theta$ by the parameter sampling rule $\Theta$, provided that $ \mathcal{S}_{\bf X} \subset K(I)\cap [ F_{\bar{S}^c_\theta}^{-1+}(0),F_{\bar{S}^c_\theta}^{-1}(1) ] $ for any $\theta\in I$. The following statement provides sufficient conditions under which the randomization approach is feasible.

    \begin{proposition}\label{eq:absolute_penalty_sufficient_con}
        Suppose that the probability measures $(\mathbb{Q}_i^\theta)_{i=1}^n$ satisfy $\mathbb{Q}_i^\theta\sim \mathbb{P}$ for all $i=1,2\dots,n$ and $\theta\in I$. Then 
         \begin{equation}
     \label{eq:risk:min}
        H_i({\bf X}) := K_i(\Theta)= \left(F^{\mathbb{Q}_i^\theta }_i\right)^{-1(\alpha_S)}(F_{\bar{S}^c_\theta}(S))\bigg|_{\theta=\Theta}, \ i=1,2,\dots,n,
    \end{equation}
    defines a risk-sharing rule.
    \end{proposition}

\begin{proof}
  See Appendix \ref{sec:app:pf:PP41}. 
\end{proof}

\begin{example}
    \label{ex:absolute:dev:quantile}
Assume that the condition in Proposition \ref{eq:absolute_penalty_sufficient_con} holds. Suppose further that $\mathbb{Q}_i$ is independent of the parameter $\theta\in I$ for all $i=1,2,\dots,n$; then, $\bar{S}^c:\equiv \bar{S}^c_\theta$, and  by Propositions 5.11 and 6.4 in \cite{dhaene:2021},  the risk-sharing rule \eqref{eq:risk:min} is given by
    \begin{equation*}
        H_i({\bf X})= \left(F^{\mathbb{Q}_i }_i\right)^{-1(\alpha_S)}(F_{\bar{S}^c}(S)) = \mathbb{E} [\bar{X}^c_i | \bar{S}^c = s  ]\bigg|_{s=S}, \ i=1,2,\dots,n,
    \end{equation*}
where $\bar{X}^c_i:= (F^{\mathbb{Q}_i}_i)^{-1}(U)$, and so $\bar{S}^c = \sum_{i=1}^n\bar{X}^c_i$.
The risk-sharing rule induced by minimizing the sum of absolute penalties is thus a generalized comonotonic CMRS rule. 

Since $\mathbb{Q}_i\sim \mathbb{P}$ for all $i=1,2,\dots,n$, the risk-sharing rule \eqref{eq:risk:min} satisfies the axiom of risk-fairness, i.e., $H_i({\bf X})\leq \mathop{\text{ess}\,\sup} X_i$ a.s. However, the actuarial-fairness might not be satisfied, as $\mathbb{E}[H_i({\bf X})] \neq \mathbb{E}[X_i]$ in general. In particular, if $\mathbb{Q}_i = \mathbb{P}$ for all $i=1,2,\dots,n$, we are able to retrieve the quantile risk-sharing rule introduced in  \cite{dhaene:2021}; this implies that the quantile risk-sharing rule, for each scenario, minimizes the sum of absolute penalties, and hence is scenario-wise Pareto optimal.

\end{example}

\subsection{Euler Principle}
\label{sec:euler}
The Euler principle allocates capital based on the marginal contribution of each individual risk to the aggregate risk measured by a prescribed risk measure. Let $(\rho_\theta)_{\theta\in I}$ be a family of positively homogeneous, with degree $1$, risk measures (i.e., $\rho_\theta(\lambda X) = \lambda \rho_\theta(X)$ for any $\lambda>0$) which are Gateaux differentiable in the interior of $I$. For the risk vector ${\bf X}$, the aggregate capital is determined by $K(\theta):=\rho_\theta(S)$, and the allocated capital based on the Euler principle is given by
    \begin{equation*}
    K_i(\theta) =     \frac{\partial}{\partial t} \rho_\theta(S+tX_i) \bigg|_{t=0}, \ i=1,2,\dots,n.
    \end{equation*}
The positive homogeneity of $(\rho_\theta)_{\theta\in I}$ ensures that $\sum_{i=1}^n K_i(\theta)=\rho_\theta(S)=K(\theta)$. 

If $\theta\in I \mapsto K(\theta) = \rho_\theta(S)$ is measurably right-invertible, there exists a random variable $\Theta$ taking values in $I$ such that $K(\Theta)=S$, and thus one can define naturally a $({\bf K},\Theta)$-induced risk-sharing rule by     \begin{equation}
    \label{eq:risk:sharing:euler}
    H_i({\bf X}) :=  K_i(\Theta)=\frac{\partial}{\partial t} \rho_\theta(S+tX_i) \bigg|_{t=0 , \theta=\Theta }, \ i=1,2,\dots,n.
\end{equation}

Inheriting from the Euler capital allocation principle, the induced risk-sharing rule is based on each participant’s marginal contribution to the aggregate risk. The key difference between the Euler principle and the induced sharing rule is that, such feature is interpreted in a scenario-wise fashion: a scenario $\omega$ is randomly realized, and each participant’s contribution is evaluated \emph{ex-post} under the corresponding risk measure $\rho_{\Theta\left(\omega\right)}$ via \eqref{eq:risk:sharing:euler}.



\begin{example} %
    The capital allocation \eqref{eq:motivating_example_allocation} presented in Example \ref{ex:motivate} belongs to the Euler principle with the aggregate capital determined by $K(\theta) = \rho_\theta(S) = \text{VaR}_\theta(S)$, for $\theta\in I=[0,1]$. With  $\Theta = F_S(S)$, we derived the CMRS rule. 
    Hence, if ${\bf X}$ admits strictly increasing marginal distribution functions for all individual risks,  the CMRS rule can be interpreted as a method which allocates the aggregate loss based on the marginal contribution of each participant as measured by the Value-at-Risk, with the confidence level depending on the realization.  
\end{example}

\begin{example}
\label{ex:exponential}
    Suppose that 
    $\mu_S\neq 0$, and the aggregate capital function is determined by 
        \begin{equation}
        \label{eq:exp:risk:measure}
          K(\theta) =  \rho_\theta(S) = \mathbb{E}\left[ S e^{ \frac{ \theta \cdot S}{\mu_S}  } \right], \ \theta\in I,
        \end{equation}
    in which we assumed that $Se^{\frac{\theta \cdot S}{\mu_S}} \in L^1(\Omega,\mathcal{F},\mathbb
    {P})$ for all $\theta\in I$. The Euler capital allocation principle is given by:
        \begin{equation}
        \label{eq:K:exp}
            K_i(\theta) = \mathbb{E}\left[X_i e^{ \frac{ \theta \cdot S}{\mu_S}  } \right] + \frac{\theta \cdot \mu_i}{\mu_S} \mathbb{E}\left[S  e^{ \frac{ \theta \cdot S}{\mu_S}  }  \left( \frac{X_i}{\mu_i} - \frac{S}{\mu_S} \right) \right], \ i=1,2,\dots,n.
        \end{equation}
   
    Suppose further that the joint distribution of ${\bf X}$ belongs to the natural exponential family with the density function:
        \begin{equation*}
            f({\bf x}; \boldsymbol{\lambda} ) = h({\bf x}) \exp( \boldsymbol{\lambda}^\top {\bf x} - A(\boldsymbol{\lambda} )  ), \ {\bf x} \in \mathbb{R}^n,
        \end{equation*}
    where $\boldsymbol{\lambda}\in\mathbb{R}^n$,  $h:\mathbb{R}^n\to \mathbb{R}_+$ determines the support of ${\bf X}$ independently of ${\bf \lambda}$,  and $A:\mathbb{R}^n\to \mathbb{R}$ is the cumulant function assumed to be differentiable.  
    The moment generating function of ${\bf X}$ is given by 
        \begin{equation*}
            M_{\bf X}({\bf t}) =e^{A(\boldsymbol{\lambda+t}) - A( \boldsymbol{\lambda} )}, \ {\bf t} \in \mathbb{R}^n. 
        \end{equation*}
    By denoting ${\bf 1}\in\mathbb{R}^n$ as the vector with all entries equal to $1$, and defining $\boldsymbol{\lambda}_\theta:= \boldsymbol{\lambda} + \theta/\mu_S {\bf 1}$, since, by the definition of moment generating function,
    \begin{equation*}
    \frac{d}{d\theta}M_{\bf X}\left(\frac{\theta}{\mu_S}{\bf 1}\right) = \frac{d}{d\theta}\mathbb{E}\left[e^{\frac{\theta}{\mu_S} {\bf 1}^\top{\bf X} } \right] = \frac{d}{d\theta}\mathbb{E}\left[e^{\frac{\theta}{\mu_S} S} \right] = \frac{1}{\mu_S}\mathbb{E}\left[Se^{\frac{\theta}{\mu_S} S} \right] = \frac{1}{\mu_S}K(\theta),
    \end{equation*}
    we have
 \begin{align*}
            K(\theta)&=\rho_\theta(S)= \mu_S \frac{d}{d\theta} M_{\bf X}\left(\frac{\theta}{\mu_S}{\bf 1}\right) = \mu_S \frac{d}{d\theta}e^{A(\boldsymbol{\lambda}_\theta)-A(\boldsymbol{\lambda})}\\
            &= e^{A(\boldsymbol{\lambda}_\theta)-A(\boldsymbol{\lambda})}\sum_{i=1}^n \partial_i A(\boldsymbol{\lambda}_\theta) =: e^{-A(\boldsymbol{\lambda})} \alpha(\theta),
        \end{align*}
     where, for $i=1,2,\dots,n$, $ \partial_i A(\cdot)$ denotes the partial derivative of $A\left(\cdot\right)$ with respect to $i$-th variable. Hence, $K:I\to\mathcal{S}_{\bf X}$ is measurably right-invertible if and only if $\alpha:I \to e^{A(\boldsymbol{\lambda})} \mathcal{S}_{\bf X}$ is  measurably right-invertible, where $e^{A(\boldsymbol{\lambda})} \mathcal{S}_{\bf X}:= \{e^{A(\boldsymbol{\lambda})} s: s\in \mathcal{S}_{\bf X}\}$. Assume that this condition holds in the example.
     
     Similarly, we also have  
        \begin{align*}
             \mathbb{E}\left[X_i e^{ \frac{ \theta \cdot S}{\mu_S}  } \right] &=e^{A(\boldsymbol{\lambda}_\theta) - A(\boldsymbol{\lambda})} \partial_iA(\boldsymbol{ \lambda}_\theta) ,\\
             \mathbb{E}\left[X_i S e^{ \frac{ \theta \cdot S}{\mu_S}  }\right] &=  e^{A(\boldsymbol{\lambda}_\theta) - A(\boldsymbol{\lambda})}\sum_{j=1}^n \left(   \partial^2_{ij} A(\boldsymbol{\lambda}_\theta) + \partial_iA(\boldsymbol{\lambda}_\theta) \partial_jA(\boldsymbol{\lambda}_\theta)  \right) ,\\
               \mathbb{E}\left[S^2 e^{ \frac{ \theta \cdot S}{\mu_S}  }\right] &=   e^{A(\boldsymbol{\lambda}_\theta) - A(\boldsymbol{\lambda})}\sum_{j,k=1}^n\left( \partial^2_{jk}A(\boldsymbol{\lambda}_\theta) + \partial_jA(\boldsymbol{\lambda}_\theta)\partial_kA(\boldsymbol{\lambda}_\theta) \right),
        \end{align*}
    where $\partial^2_{ij}A\left(\cdot\right) = \partial_i(\partial_j A\left(\cdot\right))$ and $\partial^2_{jk}A\left(\cdot\right) = \partial_j(\partial_k A\left(\cdot\right))$.
        
    By the measurable right-invertibility of $K:I\to \mathcal{S}_{\bf X}$, define  $\Theta := \alpha^{-1}(Se^{A( \boldsymbol{\lambda} )})$ so that  $K(\Theta)=\rho_\Theta(S) = S$. Therefore, the induced risk-sharing rule by the capital allocation principle \eqref{eq:K:exp} with the aggregate capital function \eqref{eq:exp:risk:measure} is given by:
        \begin{equation*}
        \begin{aligned}
            H_i({\bf X}) =&\ \Bigg\{  \partial_iA(\boldsymbol{\lambda}_\Theta)   +  \frac{\Theta \cdot \mu_i }{\mu_S}   \Bigg( \frac{\sum_{j=1}^n \left( \partial_{ij}A(\boldsymbol{\lambda}_\Theta ) +  \partial_iA(\boldsymbol{\lambda}_\Theta )\partial_jA(\boldsymbol{\lambda}_\Theta ) \right) }{\mu_i} \\
            &\ \quad - \frac{\sum_{j,k=1}^n\left( \partial^2_{jk}A(\boldsymbol{\lambda}_\Theta )+\partial_jA(\boldsymbol{\lambda}_\Theta )\partial_kA(\boldsymbol{\lambda}_\Theta ) \right)}{\mu_S} \Bigg) \Bigg\}\frac{S}{\sum_{j=1}^n \partial_jA(\boldsymbol{\lambda}_\Theta)}, \ i=1,2,\dots,n.
            \end{aligned}
        \end{equation*}
\end{example}

\begin{example}
    Continuing Example \ref{ex:exponential}, 
    suppose further that ${\bf X}$ follows a multivariate normal distribution with the density function: for ${\bf x}\in\mathbb{R}^n$,
    \begin{equation*}
        f({\bf x}) = \frac{1}{(2\pi|{\bf \Sigma}|)^{\frac{n}{2}}}\exp\left( -\frac{1}{2} ({\bf x} -\boldsymbol{\mu})^\top {\bf\Sigma}^{-1} ({\bf x} - \boldsymbol{\mu}) \right),
    \end{equation*}
where ${\bf \Sigma}$ is positive definite. Let $\boldsymbol{\lambda} := {\bf\Sigma}^{-1} \boldsymbol{\mu} $. Then, the corresponding cumulant function is given by $A(\boldsymbol{\lambda}) = \boldsymbol{\lambda}^\top {\bf\Sigma} \boldsymbol{\lambda}/2$, whence, for $\theta\in I=\mathbb{R}$,
    \begin{align*}
        \partial_iA({ \boldsymbol{\lambda}_\theta} ) &= \mu_i + \frac{\theta \sigma_{i,S} }{\mu_S}, \ 
        \partial^2_{jk}A(\boldsymbol{\lambda}_\theta)  = \sigma_{j,k},\
        A(\boldsymbol{\lambda}_\theta) -A(  \boldsymbol{\lambda}) = \theta + \frac{\theta^2\sigma_S^2}{2\mu_S^2},
    \end{align*}
and  
    \begin{equation*}
      K(\theta)=  \rho_\theta(S) = \left( \mu_S + \frac{\theta\sigma_S^2 }{\mu_S} \right) e^{ \frac{1}{2\sigma_S^2} \left( \mu_S + \frac{\theta\sigma_S^2 }{\mu_S} \right)^2 - \frac{\mu_S^2}{2\sigma_S^2}  }. 
    \end{equation*}

    By the assumption that $\mu_S\neq 0$, it is clear that $K : \mathbb{R}\to \mathbb{R}$ is bijective so that the inverse $K^{-1}\left(\cdot\right)$ uniquely exists. By letting $\Theta:=K^{-1}(S)$, we have $\rho_\Theta(S)=S$. 
More explicitly, we have 
    \begin{equation*}
        \Theta = \frac{\mu_S}{\sigma_S}\sqrt{ W_0\left( \frac{S^2}{\sigma_S^2}e^{ \frac{\mu_S^2}{\sigma_S^2}} \right) } - \frac{\mu_S^2}{\sigma_S^2},
    \end{equation*}
where $W_0(\cdot)$ is the Lambert W function of order zero. Therefore, the induced risk-sharing rule by the capital allocation principle \eqref{eq:K:exp} with the aggregate capital function \eqref{eq:exp:risk:measure} is given by: for $i=1,2,\dots,n$, 
    \begin{align*}
        H_i({\bf X}) =&\ \frac{S\left[ \mu_i + \frac{\Theta\sigma_{i,S}}{\mu_S} + \frac{\mu_i\Theta }{\mu_S}  \left( \frac{\sigma_{i,S} + \left(\mu_i + \frac{\Theta \sigma_{i,S}}{\mu_S} \right)\left(\mu_S + \frac{\Theta \sigma_S^2}{\mu_S} \right) }{\mu_i}  - \frac{\sigma_S^2 + \left(\mu_S + \frac{\Theta \sigma_S^2}{\mu_S} \right)^2 }{\mu_S}  \right) \right]}{\mu_S + \frac{\Theta\sigma_S^2}{\mu_S}} \\
        &= \frac{S\left[\mu_i + \frac{\Theta\sigma_{i,S}}{\mu_S} + \frac{\Theta}{\mu_S^2}\left(\mu_S\sigma_{i,S} - \mu_i\sigma_S^2 \right)\left(1+ \frac{\Theta\left(\mu_S+ \frac{\Theta \sigma_S^2}{\mu_S}\right) }{\mu_S} \right) \right]}{\mu_S + \frac{\Theta\sigma_S^2}{\mu_S}}. 
    \end{align*}
    
\end{example}

\subsection{Aggregate Capital Based on Distortion Risk Measures}
\label{sec:distortion}
Distortion risk measures form an important class of risk measures, encompassing many well-known examples in practice such as VaR and TVaR. They satisfy a number of desirable properties such as monotonicity, positive homogeneity, law invariance, translation invariance, and comonotonic additivity.  Each distortion risk measure is characterized by an underlying distortion function.
Thus, properties, such as the measurable right-invertibility, of the aggregate capital function, and the effects by the randomization, can be studied via examining the corresponding family of distortion functions, when the aggregate capitals are calculated by distortion risk measures on the aggregate risk.

    In this section, we let $I = [a,b]:= \{x\in \mathbb{R}: a\leq x\leq b\}$, where $a,b\in\mathbb{R}\cup\{\pm\infty\}$, $a<b$, and $I^\circ = (a,b)$ be the interior of $I$. We consider a family of distortion risk measures $(\rho_\theta)_{\theta\in I^\circ}$ given by
      \begin{equation}
            \label{eq:distortion}
            K(\theta)=\rho_\theta(S) := -\int_{-\infty}^0\left(1-D_\theta(\bar{F}_S(s)) \right)\,ds + \int_0^\infty D_\theta(\bar{F}_S(s)) \,ds,  \ \theta\in I^\circ,
      \end{equation}       
    where, for $\theta\in I^\circ$, $D_\theta(\cdot) :[0,1]\to[0,1]$ is a distortion function, i.e., $D_\theta$ is non-decreasing  and satisfies $D_\theta(0)=0$, $D_\theta(1)=1$; by definition,  $K(\theta) \in \mathbb{R}$ if both integrals in \eqref{eq:distortion} are finite. 
    We also define 
        \begin{equation*}
        K(a)=\rho_a(S) :=    \lim_{\theta\to a^+} \rho_\theta(S), \text{ and } K(b)=\rho_b(S) :=\lim_{\theta\to b^-} \rho_\theta(S),
        \end{equation*}
   where $\rho_a(S),\rho_b(S)\in \mathbb{R}\cup\{\pm\infty\}$. We thus extend the family by including the two end-points, yielding the extended family $(\rho_\theta)_{\theta\in I}$.
    
\subsubsection{Measurably right-invertible family of distortion risk measures}
 \label{sec:aggregate:distortion}

   Proposition \ref{pp:distortion} below presents a sufficient condition on the family of distortion functions $(D_\theta)_{\theta\in I^\circ}$ such that the associated distortion risk measures $(\rho_\theta)_{\theta\in I}$ (including the limits at the two end-points) is non-decreasing and surjective onto $[F_S^{-1+}(0),F_S^{-1}(1)]$. By this and Proposition \ref{pp:K:surjective}, the associated aggregate capital function $K(\cdot)$ is measurably right-invertible. To summarize, the conditions in Proposition \ref{pp:distortion} imply that $(D_\theta)_{\theta\in I^\circ}$ is a continuous deformation from $\mathbbm{1}_{\{1\}}(\cdot)$ to $\mathbbm{1}_{(0,1]}(\cdot)$, so that the integral  in \eqref{eq:distortion} can sweep through the entire support $\mathcal{S}_{\bf X}$.

        \begin{proposition} 
        \label{pp:distortion}  
            Let $(D_\theta)_{\theta\in (a,b)}$ be a family of distortion functions and $(\rho_\theta)_{\theta\in (a,b)}$ be the associated family of distortion risk measures defined in \eqref{eq:distortion}.  Suppose that for any $p\in[0,1]$, $\theta \in (a,b)\mapsto D_\theta(p)$ is non-decreasing. Then, $\theta \in (a^*,b^*) \mapsto K(\theta)$ is non-decreasing, where $a^*:=\inf\{\theta\in(a,b):\rho_\theta(S)>-\infty\}$, $b^*:= \sup\{\theta\in(a,b):\rho_\theta(S)<\infty\}$, and $a^*\leq b^*$. 

            Moreover, suppose that the following conditions hold: for all $p\in[0,1]$,
                \begin{enumerate}[label=(\alph*)]
                    \item  $\theta\mapsto D_\theta(p)$ is continuous on $(a,b)$;
                    \item $\lim_{\theta\downarrow a} D_\theta(p) = \mathbbm{1}_{\{1\}}(p)$;
                    \item $\lim_{\theta\uparrow b}D_\theta(p) = \mathbbm{1}_{(0,1]}(p)$.
                \end{enumerate}
Then, the following statements hold:
\begin{enumerate}
    \item \( K:(a^*, b^*) \to \mathbb{R} \) is continuous;
    \item \( \lim_{\theta \downarrow a} K(\theta) = F_S^{-1+}(0) \), and \( \lim_{\theta \uparrow b} K(\theta) = F_S^{-1}(1) \);\footnote{
The limits are understood in $\mathbb{R}\cup\{\pm\infty\}$. 
If $a^*>a$ (resp.~$b^*<b$), then $F_S^{-1+}(0)=-\infty$
(resp.~$F_S^{-1}(1)=+\infty$), and so
$K(\theta)=-\infty$ for $\theta\in(a,a^*]$ (resp.~$K(\theta)=+\infty$ for $\theta\in[b^*,b)$.}

    \item \( K:(a^*,b^*)\to\mathbb{R} \) is surjective onto \( (F_S^{-1+}(0), F_S^{-1}(1)) \);

\end{enumerate} 
  \end{proposition}
        \begin{proof}
        See Appendix \ref{sec:app:pf:PP42}. 
        \end{proof}
    
    The conditions in Proposition \ref{pp:distortion} offer a natural and tractable setting under which the results hold, but is not necessary for measurable right-invertibility. For example, the family $(\text{VaR}_\theta(Z))_{\theta\in[0,1]}$ is clearly non-decreasing and surjective onto the support of a generic random variable $Z$. However, the associated family of distortion functions $( \mathbbm{1}_{(1-\theta,1]}(\cdot) )_{\theta\in[0,1]}$ does not satisfy Condition (a) in Proposition \ref{pp:distortion}. Many widely used distortion risk measures, such as those based on Wang's transformation, proportional hazards, or dual power distortions, do satisfy all the conditions listed, ensuring that the proposition is widely applicable in practice; the following examples illustrate these.
   
    \begin{example}\label{eg:wang_example}
        The Wang's transformation $D_\theta(p) = \Phi(\Phi^{-1}(p)-\Phi^{-1}(1-\theta))$, for $\theta\in(0,1)$ and $p\in[0,1]$, satisfies the conditions in Proposition \ref{pp:distortion}, where $\Phi(\cdot)$ denotes the distribution function of the standard normal random variable. 
    \end{example}   
   
     \begin{example}
         The power distortion function $D_\theta(p) = p^{\frac{1-\theta}{\theta}}$, for $\theta\in(0,1)$ and $p\in[0,1]$, satisfies the conditions in Proposition \ref{pp:distortion}.
      \end{example}

    As mentioned earlier, concave distortion functions alone cannot form a family of measurably right-invertible distortion risk measures since $\rho_\theta(S)\geq \mathbb{E}[S]$. Nevertheless, by combining  concave distortions with their dual, one can obtain a measurably right-invertible family; the following example demonstrates this.
     \begin{example}
     \label{ex:dual}
            Let $(D_\theta)_{\theta\in[0,1)}$ be a family of distortions which satisfies Condition (a) in Proposition \ref{pp:distortion}, with $D_0(p)=p$ and $\lim_{\theta\to 1^-} D_\theta(p) = \mathbbm{1}_{(0,1]}(p)$, for $p\in[0,1]$. Define the dual\footnote{The dual herein is defined by also reversing the parameter from $\theta$ to $1-\theta$, different from the usual definition of dual in the literature.} of $D_\theta$ by $D^*_\theta(p) := 1- D_{1-\theta}(1-p)$, for $p\in[0,1]$. Let, for $p\in[0,1]$,
                \begin{equation*}
                    \tilde{D}_\theta(p) := \begin{dcases}
                       D^*_{2\theta}(p), &\text{if } \theta \in\left(0,\frac{1}{2}\right]; \\
                       D_{2\theta-1}(p), &\text{if }\theta \in \left( \frac{1}{2},1 \right).
                    \end{dcases}
                \end{equation*}
            Then, the family of distortion risk measures $(\rho_\theta)_{\theta\in[0,1]}$ associated with $(\tilde{D}_\theta)_{\theta\in[0,1]}$ satisfies the conditions in Proposition \ref{pp:distortion}. 
    \end{example}
   
    In Example \ref{ex:dual}, the first component $D_{2\theta}^*$ reflects and scales the original family of distortions $(D_\theta)_{\theta\in[0,1)}$ about the identity line. This enables the associated distortion risk measure $\rho_\theta(S)$ to attain values on $(F_S^{-1+}(0),\mu_S]$. The second component $D_{2\theta-1}(p)$ builds on this by scaling and translating the original distortions so that $\rho_\theta(S)$ can attain values greater than $\mu_S$.
    
    The following example complements the distortion function of TVaR by its dual. The resulting risk measure assigns more weights on both tails of the distribution, providing a balanced assessment of deviations from central outcomes.


\begin{example}
    Let $(D_\theta)_{\theta\in [0,1)}$ be a family of distortions defined by $D_\theta(p):=\frac{p}{1-\theta}\wedge 1$, for any $p\in [0,1]$ and $\theta\in[0,1)$. Following the definition in Example \ref{ex:dual}, for $p\in[0,1]$, define the family of distortions $(\tilde{D}_\theta)_{\theta\in(0,1)}$ by  
        \begin{equation*}
            \tilde{D}_\theta(p) := \begin{dcases}
                \frac{1}{2\theta}\left(p-(1-2\theta)\right)_+ ,  &\text{if } \theta \in\left(0,\frac{1}{2}\right]; \\
                \min\left\{\frac{p}{2(1-\theta)},1 \right\} , &\text{if } \theta \in\left(\frac{1}{2},1 \right).
            \end{dcases}
        \end{equation*}
       Then, the conditions in Proposition \ref{pp:distortion} are satisfied. The family of 
       distortion risk measures $(\rho_\theta(S))_{\theta\in (0,1)}$ associated with $(\tilde{D}_\theta)_{\theta\in (0,1)}$ is given by 
            \begin{align*}
                \rho_\theta(S) = \begin{dcases}
          \frac{1}{2\theta}\int_0^{2\theta} F_S^{-1}(p)\,dp , &  \text{if } \theta \in\left(0,\frac{1}{2}\right];\\
                    \text{TVaR}_{2\theta-1}(S) , & \text{if } \theta \in\left(\frac{1}{2},1\right).
                \end{dcases}
            \end{align*}
\end{example}

\subsubsection{Euler principle based on distortion risk measures}\label{sec:euler_distortion}    
Resuming the discussion in Section \ref{sec:euler}, this subsection discusses the risk-sharing rules induced by the Euler capital allocation principle, when the aggregate capitals are given by a family of distortion risk measures on the aggregate risk. We shall also discuss the connection of the comonotonicity of the induced risk-sharing rule with that of the CMRS rule in Proposition \ref{pp:distortion:euler:comono} below. 


Suppose that the associated distortion functions $(D_\theta)_{\theta\in I^\circ}$ of the risk measures $(\rho_\theta)_{\theta\in I^\circ}$ are left-continuous with respect to $p\in[0,1]$. By Theorem 6 in \cite{dhaene2012remarks}, for $\theta\in I$,
\begin{equation*}
    \rho_\theta(S) = \int_0^1\text{VaR}_{1-p}(S)\,dD_\theta(p). 
\end{equation*}
Under additional assumptions on the regularity of the conditional density functions of $X_i$, for $i=1,\dots,n$, given $S$ (see, e.g., Appendix A in \cite{TSANAKAS2003239}), the Gateaux derivative of VaR exists and admits the following expression: for $i=1,\dots,n$ and $p\in(0,1)$, 
    \begin{equation}
    \label{eq:var:derivative}
        \frac{d}{dt}\text{VaR}_p(S+tX_i)\bigg|_{t=0} = \mathbb{E}[X_i|S=\text{VaR}_p(S)].
    \end{equation}
Assume further that the following regularity condition holds: there exists $\varepsilon>0$ and a function $g:[0,1]\to[0,\infty)$ such that, for any $|t|\leq \varepsilon$ and almost all $p\in[0,1]$, and for $\theta\in I^o$,
    \begin{equation}
        \label{eq:ass:g:VaR}
       \left| \frac{\text{VaR}_{1-p}(S+tX_i)-\text{VaR}_{1-p}(S)}{t}\right| \leq g(p) \quad \text{and} \quad  \int_0^1g(p)\,dD_\theta(p) < \infty. 
    \end{equation}
By \eqref{eq:var:derivative} and the Dominated Convergence Theorem, the allocated capitals under the Euler principle  is given by, for $i=1,\dots,n$ and $\theta\in I$,
    \begin{equation}
        \label{eq:allocation:distortion}
        K_i(\theta) = \frac{d}{dt}\int_0^1 \text{VaR}_{1-p}(S+tX_i)\,dD_\theta(p)\bigg|_{t=0}  
        = \int_0^1 \mathbb{E}[X_i|S=\text{VaR}_{1-p}(S)] \,dD_\theta(p). 
        \end{equation}

Now, suppose that the family $(D_\theta)_{\theta\in (a,b)}$ satisfies the conditions in Proposition \ref{pp:distortion}. Then, the allocation  \eqref{eq:allocation:distortion} induces the following risk-sharing rule:
\begin{equation}
\label{eq:euler:distortion}
        H_i({\bf X}) = K_i(\Theta)=   \int_0^1 \mathbb{E}[X_i|S=\text{VaR}_{1-p}(S)] \,dD_\theta(p) \bigg|_{\theta=\Theta}, \ i=1,2,\dots,n.
    \end{equation}
 
\begin{example}
    \label{ex:distortion:elliptical}
    Let ${\bf X} \sim E_n(\boldsymbol{\mu},{\bf \Sigma},g)$ follows an elliptical distribution, and $(D_\theta)_{\theta\in I^\circ}$ be a family of left-continuous distortion functions that  satisfies the conditions in Proposition \ref{pp:distortion} and \eqref{eq:ass:g:VaR}. For any $s\in \mathcal{S}_{\bf X}$,
        \begin{equation*}
            \mathbb{E}[X_i|S=s] = \mu_i + \frac{\sigma_{i,S}}{\sigma_S^2}(s-\mu_S). 
        \end{equation*}
Using this and \eqref{eq:euler:distortion}, the risk-sharing rule induced by the Euler principle is given by: for $i=1,\dots,n$,
 \begin{align}
        \label{eq:elliptic:distortion}
            H_i({\bf X}) &=  \int_0^1\left( \mu_i + \frac{\sigma_{i,S}}{\sigma_S^2}\left(\text{VaR}_{1-p}(S)-\mu_S\right)\right)\,dD_\theta(p) \bigg|_{\theta=\Theta} \nonumber \\
            &= \mu_i + \frac{\sigma_{i,S}}{\sigma_S^2}\left( \int_0^1 \text{VaR}_{1-p}(S)\,dD_\theta(p)\bigg|_{\theta=\Theta} -\mu_S \right) \nonumber \\
            &= \mu_i + \frac{\sigma_{i,S}}{\sigma_S^2}\left(  \rho_\Theta(S) -\mu_S \right) \nonumber \\
            &=  \mu_i + \frac{\sigma_{i,S}}{\sigma_S^2}\left(  S -\mu_S \right).
        \end{align}

\end{example}

  From \eqref{eq:elliptic:distortion}, we see that the Euler principle always leads to the same quota-sharing rule when ${\bf X}$ follows an elliptical distribution and the aggregate capital is computed using a distortion risk measure on the aggregate risk. This risk-sharing rule is independent of the parameter sampling rule $\Theta$ in which the risk measure is evaluated. Compared to the quota-sharing rule in \eqref{eq:opt:square}, the participation coefficient on $S$ in \eqref{eq:elliptic:distortion} for the $i$-th participant  depends on the correlation between her own risk and the aggregate risk. This is a consequence of the very definition of the associated Euler principle, which computes the shared risk based on the marginal effect of $X_i$ on the risk measure of the aggregate risk $S$. 

  Another implication from \eqref{eq:elliptic:distortion} is that the induced risk-sharing rule is comonotonic if $\sigma_{i,S}\geq 0$ for all $i=1,\dots,n$. This observation can be generalized as follows.

  \begin{proposition}
  \label{pp:distortion:euler:comono}
      Suppose that ${\bf X}$ admits a joint density function on $\mathbb{R}^n$, and $S$ admits a strictly increasing distribution function.\footnote{This condition can be weakened to those stated in  Appendix A of \cite{TSANAKAS2003239}, along with the left-continuity of the mapping $s\mapsto\mathbb{E}[X_i|S=s]$ for all $i=1,\dots,n$. These conditions, along with the non-decreasing assumption of $s\mapsto \mathbb{E}[X_i|S=s]$, ensure that (i) the Gateaux derivative of VaR is given by the conditional expectation, and (ii) the mapping $p\mapsto \mathbb{E}[X_i|S=\text{VaR}_p(S)]$ is left-continuous.}  Let $(D_\theta)_{\theta\in I^\circ}$ be a family of distortion functions such that 
        \begin{enumerate}[label=(\alph*)]
        \item Condition \eqref{eq:ass:g:VaR} is satisfied for $\theta\in I^\circ$;
            \item $p\in[0,1]\mapsto D_\theta(p)$ is left-continuous for all $\theta\in I^\circ$;
            \item $\theta\in I^\circ\mapsto D_\theta(p)$ is non-decreasing for all $p\in[0,1]$;
            \item the aggregate capital function $K(\theta)=\rho_\theta(S)$ is measurably right-invertible. 
        \end{enumerate}
    Assume further that   the mapping $s\mapsto \mathbb{E}[X_i|S=s]$ is non-decreasing for all $i=1,\dots,n$. Then, the risk-sharing rule induced by the Euler principle with respect to the family of distortion risk measures $(\rho_\theta)_{\theta\in I}$ associated with $(D_\theta)_{\theta\in I^\circ}$ is comonotonic. 
  \end{proposition}

\begin{proof}
    See Appendix \ref{sec:app:pf:PP43}. 
\end{proof}

The condition $s\mapsto \mathbb{E}[X_i|S=s]$ is non-decreasing for all $i=1,\dots,n$, is well-documented in literature (see, e.g.~\cite{DENUIT2012265}), which ensures that the shared risk under the CMRS rule is comonotonic. In particular, for elliptical distributions, this condition holds if and only if $\sigma_{i,S} \geq 0$ for all $i=1,\dots,n$, which verifies the finding in Example \ref{ex:distortion:elliptical}.


\section{Randomization of Bottom-Up Principles}
\label{sec:bottom-up}

In bottom-up capital allocation principles, the aggregate capitals are determined endogenously. For a given family of bottom-up principles $\mathbf{K}(\theta)$, for $\theta\in I$, we can represent the aggregate capital function as $K(\theta)=\rho_\theta({\bf X})$, where $\rho_\theta(\cdot)$ is determined by the underlying allocation mechanism in the parametric family. In general, $\rho_\theta(\cdot)$ sometimes depends on the entire joint distribution of ${\bf X}$, rather than only on the aggregate risk $S$. In this section, we shall illustrate the randomization of the bottom-up capital allocation principles, to induce the corresponding risk-sharing rules, for two principles, namely, the weighted risk allocation principle, and the holistic principle.

\subsection{Weighted Risk Allocation Principle}
\label{sec:weighted:risk}
Throughout this section, assume that $S\geq 0$ a.s.. The weighted risk allocation principle, introduced in \cite{FURMAN2008263}, is motivated by the weighted premium calculation principles in \cite{FURMAN2008459}. This principle allocates capital according to the weighted expected risks of the LOBs.

For $\theta \in I^\circ:=(a,b)$, where $-\infty\leq a<b\leq +\infty$, let $w_\theta: [0,\infty) \to [0,\infty)$ be a weighting function. The weighted risk capital allocation is defined by: for $\theta \in I^\circ$, and for $i = 1, \dots, n$,
\begin{equation*}
    K_i(\theta) = \frac{\mathbb{E}[X_i w_\theta(S)]}{\mathbb{E}[w_\theta(S)]},
\end{equation*}
in which we assumed that $\mathbb{E}[w_\theta(S)]\neq 0$. As a bottom-up principle, the weighted risk capital allocation principle does not allocate an exogenous aggregate capital. Instead, the aggregate capital is endogenously obtained by summing the individual allocations up; i.e., for $\theta\in (a,b)$,
    \begin{equation}
    \label{eq:K:weighted}
        K(\theta) = \sum_{i=1}^n K_i(\theta) = \frac{\mathbb{E}[Sw_\theta(S)]}{\mathbb{E}[w_\theta(S)]},
    \end{equation}
which depends only on the distribution of $S$; also, define $K(a):=\lim_{\theta\downarrow a}K(\theta)$ and $K(b):=\lim_{\theta\uparrow b}K(\theta)$.

The weighted risk capital allocation principle encompasses several well-known allocation principles. It covers the Euler allocation under distortion risk measures when $w_\theta(s) = D_\theta'(\bar{F}_S(s))$, for $s\geq 0$, with a distortion function $D_\theta$, and it generates the Conditional Tail Expectation (CTE) allocation when $w_\theta(s) = \mathbbm{1}_{\{s \ge \mathrm{VaR}_\theta(S)\}}$, for $s\geq 0$.

Suppose that $K(\cdot)$ in \eqref{eq:K:weighted} is measurably right-invertible, and thus there exists a random variable $\Theta$ taking values in $I$ such that $K(\Theta)=S$.
Define an $\left(\mathbf{K},\Theta\right)$-induced risk-sharing rule by 
    \begin{equation}
    \label{eq:weighted:risk:sharing}
        H_i({\bf X}) := K_i(\Theta) =  \frac{\mathbb{E}[X_i w_\theta(S)]}{\mathbb{E}[w_\theta(S)]}\bigg|_{\theta=\Theta}.
    \end{equation}
In the following, we provide conditions on the family of the weighting functions so that $K(\cdot)$ in \eqref{eq:K:weighted} is measurably right-invertible. We begin by the following lemma. 

\begin{lemma}
\label{lem:K:monotonic:weighted}
Suppose that the family $(w_\theta(\cdot))_{\theta\in I^\circ}$ satisfies the monotone likelihood ratio property: for any $s_1,s_2\in \mathcal{S}_{{\bf X}}$ with $s_2\geq s_1$, and $\theta_1,\theta_2 \in I^\circ$ with $\theta_2\geq \theta_1$, 
 \begin{equation}
        \label{eq:cond:K:weighted:w}
        \frac{w_{\theta_2}(s_2)}{w_{\theta_1}(s_2)} \geq \frac{w_{\theta_2}(s_1)}{w_{\theta_1}(s_1)}. 
    \end{equation}
 Then, for any non-decreasing function $g(\cdot)$, the mapping 
    \begin{equation*}
        \theta\mapsto \frac{\mathbb{E}[g(S)w_\theta(S)]}{\mathbb{E}[w_\theta(S)]}
    \end{equation*}
is non-decreasing for those $\theta\in I^\circ$ such that $\mathbb{E}[|g(S)|w_\theta(S)]<\infty$ and $0<\mathbb{E}[w_\theta(S)]<\infty$. In particular, $K(\theta)$ in \eqref{eq:K:weighted} is non-decreasing in $\theta\in I^\circ$.
\end{lemma}

\begin{proof}
    See Appendix \ref{sec:app:pf:lem51}. 
        \end{proof}

Define the effective domain $(a,b^*) \subseteq I$ by,  
    \begin{equation*}
        b^* := \sup\{\theta\in (a,b) : K(\theta)<\infty \}.
    \end{equation*}
When \eqref{eq:cond:K:weighted:w} holds, by Lemma \ref{lem:K:monotonic:weighted}, $K(\theta)<\infty$ for all $\theta \in (a,b^*)$.


\begin{proposition}  
\label{pp:weighted:surjective}
      Suppose that the family $(w_\theta(\cdot))_{\theta \in I^\circ}$ satisfies \eqref{eq:cond:K:weighted:w}, and  for any $s > F_S^{-1+}(0)$ and $s'< F_S^{-1}(1)$,
        \begin{equation}
        \label{eq:cond:w:limit}
         \lim_{\theta\to a^+} \frac{\mathbb{E}[w_\theta(S) \mathbbm{1}_{\{S > s\}
         } ]}{\mathbb{E}[w_\theta(S) \mathbbm{1}_{\{S \leq  s\}
         } ]}   = \lim_{\theta\to b^-} \frac{\mathbb{E}[w_\theta(S)\mathbbm{1}_{\{S \leq s' \} }]}{\mathbb{E}[w_\theta(S)\mathbbm{1}_{\{S >  s' \} }]}= 0. 
        \end{equation}
    Then, $\lim_{\theta\downarrow a}K(\theta)=F_S^{-1+}(0)$ and $\lim_{\theta\uparrow b}K(\theta) = F_S^{-1}(1)$.\footnote{If $F_S^{-1}(1)=\infty$, then $K(\theta)=\infty$ for $\theta\in[b^*,b)$.} 
    
    In addition, if for any $s\in (F_S^{-1+}(0),F_S^{-1}(1))$, $\theta \in I^\circ \mapsto w_\theta(s)$ is continuous, and for any open interval $J\subseteq (a,b^*)$, there exists a random variable $W_J\geq 0$ such that for any $\theta\in J$,
    \begin{equation}
    \label{eq:cond:w:integrability}
        |w_\theta(S)| \leq W_J,\  \mathbb{E}[SW_J]<\infty,   \text{ and }  0<\mathbb{E}[W_J]< \infty. 
    \end{equation}
     Then, $K:(a,b^*)\to\mathbb{R}$ is continuous and surjective onto $(F_S^{-1+}(0),F_S^{-1}(1))$.
    \end{proposition}

\begin{proof}
    See Appendix \ref{sec:app:pf:PP51}. 
\end{proof}

By Proposition \ref{pp:K:surjective}, when the conditions in Lemma \ref{lem:K:monotonic:weighted} and Proposition \ref{pp:weighted:surjective} are satisfied, $K(\cdot)$ in \eqref{eq:K:weighted} is measurably right-invertible, and thus \eqref{eq:weighted:risk:sharing} is a legitimate $\left(\mathbf{K},\Theta\right)$-induced risk-sharing rule. We provide some examples for the weighting functions $(w_\theta(\cdot))_{\theta\in I^\circ}$ that verify these conditions.

\begin{example}[Size-biased allocation]
\label{ex:size:biased}
    Let $I= \mathbb{R}$  and consider, for $\theta\in I$, $w_\theta(s) = s^\theta$, for $s\geq 0$. Assume that there exists $p>0$ such that  $0<\mathbb{E}[S^\theta]<\infty$ for all $|\theta|\leq p$. Note that this assumption implies \eqref{eq:cond:w:integrability}. In addition, \eqref{eq:cond:K:weighted:w} is clearly satisfied, and $\theta\mapsto w_\theta(s)$ is continuous for all $s\geq 0$.  

    Next, for any $s < F_S^{-1}(1)$, there exists $\varepsilon>0$ such that 
    $s+\varepsilon < F_S^{-1}(1)$. Hence, $\mathbb{P}(S>s+\varepsilon)>0$, and for $\theta>0$,
        \begin{align*}
            \frac{\mathbb{E}[S^\theta\mathbbm{1}_{\{S \leq s \} } ]}{\mathbb{E}[S^\theta\mathbbm{1}_{\{S > s \} } ]} \leq  \frac{s^\theta}{\mathbb{E}[S^\theta\mathbbm{1}_{\{S > s +\varepsilon \} } ]} \leq \left(\frac{s}{s+\varepsilon}\right)^\theta \frac{1}{\mathbb{P}(S>s+\varepsilon)} \to 0
        \end{align*}
    as $\theta \to +\infty$.
    
    Likewise, for any  $s>F_S^{-1+}(0)$, there exists $\varepsilon>0$ such that $s-\varepsilon>F_S^{-1+}(0)$. Hence, $\mathbb{P}(S\leq s-\varepsilon)>0$, and  for $\theta<0$, 
        \begin{equation*}
           \frac{\mathbb{E}[S^\theta\mathbbm{1}_{\{S > s \} } ]}{\mathbb{E}[S^\theta\mathbbm{1}_{\{S \leq s \} } ]}  \leq \frac{s^\theta}{\mathbb{E}[S^\theta \mathbbm{1}_{\{S \leq s-\varepsilon \}}]} \leq \left(\frac{s}{s-\varepsilon} \right)^\theta \frac{1}{\mathbb{P}(S\leq s-\varepsilon)} \to 0
        \end{equation*}
    as $\theta \to -\infty$. Therefore, \eqref{eq:cond:w:limit} is also satisfied.

\end{example}

\begin{example}[Esscher's transform]
    Let $I=\mathbb{R}$ and  consider, for $\theta\in I$, $(w_\theta(\cdot))_{\theta\in I^\circ}$, for $s\geq 0$. Assume that there exists $p>0$ such that $\mathbb{E}[S e^{\theta S}]<\infty$ and $0<\mathbb{E}[e^{\theta S}] < \infty$ for any $|\theta| \leq p$. The conditions \eqref{eq:cond:K:weighted:w}, \eqref{eq:cond:w:limit}, and \eqref{eq:cond:w:integrability} can be verified in a similar manner as in Example \ref{ex:size:biased}. 
\end{example}

The weighted risk allocation principle can be characterized by the solution to the following risk-weighted minimization problem: for $\theta\in I$ and $i=1,\dots,n$,
\begin{equation*}
    K_i(\theta) = \mathop{\arg\min}_{K_i\in\mathbb{R}} \mathbb{E}\Big[ (X_i-K_i)^2 w_\theta(S) \Big]. 
\end{equation*}
Consequently, the induced risk-sharing rule \eqref{eq:weighted:risk:sharing} inherits this property in a scenario-wise manner, analogous to the optimization principle discussed in Section \ref{sec:optim}. However, unlike the top-down principles considered therein, the aggregate capital function and the random parameter $\Theta$ are determined endogenously by the allocation rule via \eqref{eq:K:weighted} in this case.

Similar to the risk-sharing rules induced by the Euler principle based on distortion risk measures, the comonotonicity of the shared risk \eqref{eq:weighted:risk:sharing} is connected to the comonotonicity of the CMRS rule. 

\begin{proposition}
    \label{pp:comono:weighted:risk}
    Suppose that the family of the weighting functions $(w_\theta(\cdot))_{\theta\in I^\circ}$ satisfies \eqref{eq:cond:K:weighted:w}, and the mapping $s\mapsto \mathbb{E}[X_i|S=s]$ is non-decreasing for all $i=1,\dots,n$. Then, the induced risk-sharing rule \eqref{eq:weighted:risk:sharing}, whenever it is well-defined, is comonotonic.
\end{proposition}

\begin{proof}
   See Appendix \ref{sec:app:pf:PP52}. 
\end{proof}

\subsection{Holistic Principle}
\label{sec:holistic} 

In this subsection, we discuss the randomization of the holistic principle, focusing on the measurable right-invertibility of the aggregate function $\rho_\theta({\bf X})$, for $\theta\in I$, and the properties of the induced risk-sharing rule.   

The holistic capital allocation approach was introduced in \cite{chong2021holistic}, which aims to find the allocation $(K_i)_{i=1}^n$ and the aggregated capital $K$ simultaneously in one single optimization problem. By minimizing the squared errors, the $(n+1)$-tuple $(K_1,\dots,K_n,K)$ is obtained by solving:
    \begin{equation}
        \label{eq:holistic}
        \inf_{\substack{K_1,\dots,K_n\\ \sum_{i=1}^nK_i=K}}  \sum_{i=1}^n \gamma_i \mathbb{E}\left[ (X_i-K_i)^2 u_i(X_i)  \right] + \gamma \mathbb{E}\left[ (S-K)^2 u(S)  \right],
    \end{equation}
where $(\gamma_i)_{i=1}^n$, $\gamma>0$ are positive weights, and $(u_i(X_i))_{i=1}^n$, $u(S)$ are non-negative random variables such that $\mathbb{E}[u_i(X_i)]=1=\mathbb{E}[u(S)]$. Herein, we assumed that $Su(S)$, $S^2u(S), X_iu_i(X_i),X_i^2u_i(X_i)\in L^1(\Omega,\mathcal{F},\mathbb{P})$ for all $i=1,\dots,n$. The solution of \eqref{eq:holistic} is given by 
\begin{equation}
\label{eq:holistic:sol}
    \begin{aligned}
        K &= \rho(S)+ \beta\left( \sum_{j=1}^n \rho_j(X_j) -\rho(S) \right)  , \  \\
        K_i &=  \rho_i(X_i) -  \beta_i \left(\sum_{j=1}^n \rho_j(X_j) -\rho(S) \right), \ i=1,2,\dots,n,
    \end{aligned}
    \end{equation}
where 
    \begin{equation}
    \label{eq:holistic:sol:rho}
        \rho(S) := \mathbb{E}[Su(S)], \ \rho_i(X_i) := \mathbb{E}[X_iu_i(X_i)], \ \beta_i := \frac{\frac{1}{\gamma_i}}{\frac{1}{\gamma} + \sum_{j=1}^n \frac{1}{\gamma_j} }, \ \beta := \frac{\frac{1}{\gamma}}{\frac{1}{\gamma} + \sum_{j=1}^n \frac{1}{\gamma_j} } .
    \end{equation}

To deduce risk-sharing rules from the holistic capital allocation principle, we consider the following parametrized family of the optimization problems \eqref{eq:holistic}:
    \begin{equation}
          \label{eq:holistic:randomized}  {\bf K}(\theta) = \mathop{\arg\min}_{K_1,\dots,K_n}  \sum_{i=1}^n \gamma_i \mathbb{E}\left[ (X_i-K_i)^2 u_{i,\theta}(X_i)  \right] + \gamma \mathbb{E}\left[ \left(S-\sum_{j=1}^nK_j\right)^2 u_\theta(S)  \right],
    \end{equation}
where $\theta\in I$, and $(u_{i,\theta}\left(\cdot\right))_{i=1}^n$, $u_\theta\left(\cdot\right)$ are parametrized functions.  If the aggregate capital function $K\left(\cdot\right):=\sum_{i=1}^nK_i\left(\cdot\right) : I\to \mathbb{R}$ is measurably right-invertible on $\mathcal{S}_{\bf X}$, then there exists a random variable $\Theta$ taking values in $I$ such that $K(\Theta)=S$.

 Heuristically, this randomized holistic principle renders an
 allocation rule ${\bf K}(\theta)$ that minimizes the risk-weighted sum of mean-squared difference for each scenario $\theta\in I$. 
Unlike the randomized top-down principles discussed earlier where  $K(\theta)$, and thus the parameter sampling rule $\Theta$, are determined by the aggregated capital function via its inverse, the randomized holistic principle determines $\Theta$ based on the underlying allocation mechanism. As a result, the random parameter in this framework depends on the distribution of the entire risk vector ${\bf X}$ and the choice of the derivative functions in the optimization \eqref{eq:holistic:randomized}; see Example \ref{ex:holistic:elliptic} for an illustration.

Let $\rho_\theta(S):=\mathbb{E}[Su_\theta(S)]$ and $\rho_{i,\theta}(X_i) := \mathbb{E}[X_iu_{i,\theta}(X_i)]$, $i=1,2,\dots,n$, which are the parametrized version of $\rho(S)$ and $\rho_i(X_i)$ in \eqref{eq:holistic:sol:rho}, respectively.  To obtain a legitimate risk-sharing rule, surjectivity conditions of individual risk measures $\rho_{i,\theta}(X_i)$, $i=1,2,\dots,n$,  in addition to $\rho_\theta(S)$, with respect to $\theta$, shall be needed. Similar to Section~\ref{sec:aggregate:distortion}, we let $I = [a, b]$ with $-\infty\leq a<b\leq \infty$, and consider the mapping $\theta \in I^\circ = (a, b) \mapsto \rho_\theta(S)$.  We also define the boundary values by
\[
\rho_a(S) := \lim_{\theta \downarrow a} \rho_\theta(S), \quad \text{and} \quad \rho_b(S) := \lim_{\theta \uparrow b} \rho_\theta(S),
\]
which take values in $\mathbb{R}\cup\{\pm\infty\}$. The mappings $\theta \in I \mapsto \rho_{i,\theta}(X_i)$, for $i = 1, \dots, n$, are defined analogously.  The following result outlines sufficient conditions for the well-posedness of this randomization approach.


    \begin{proposition}
    \label{pp:holistic}
      Suppose that all of the following conditions hold:
        \begin{enumerate}[label=(\alph*)]
            \item $\theta\in I \mapsto \rho_\theta(S)$ is continuous on $I^\circ$ and surjective onto $[F_S^{-1+}(0),F_S^{-1}(1)]$;
            \item $\theta \in I\mapsto \rho_{i,\theta}(X_i) $ is continuous on $I^\circ$ and surjective onto $[F_{X_i}^{-1+}(0),F_{X_i}^{-1}(1)]$ for $i=1,\dots,n$;
            \item   $\lim_{\theta\to a^+} \rho_\theta(S) =F_S^{-1+}(0)$, $\lim_{\theta\to b^-} \rho_\theta(S) =F_S^{-1}(1)$ and $\lim_{\theta\to a^+} \rho_{i,\theta}(X_i) =F_{X_i}^{-1+}(0)$, $\lim_{\theta\to b^-} \rho_{i,\theta}(X_i) =F_{X_i}^{-1}(1)$ for $i=1,2,\dots,n$.
        \end{enumerate}
        Then,
            \begin{equation}
            \label{eq:holistic:risk}
                  H_i({\bf X}) = K_i(\Theta) :=  \rho_{i,\Theta}(X_i) -  \beta_i \left(\sum_{j=1}^n \rho_{j,\Theta}(X_j) -\rho_\Theta(S) \right), \ i=1,2,\dots,n
            \end{equation}
        is an $\left(\mathbf{K},\Theta\right)$-induced risk-sharing rule.
    \end{proposition}
    \begin{proof}
      See Appendix \ref{sec:app:pf:PP53}. 
    \end{proof}

\begin{example}
  Let $(D_\theta)_{\theta \in I^\circ}$ be a family of distortion functions satisfying the conditions in Proposition \ref{pp:distortion}, and  $D_\theta$ is differentiable for each $\theta \in I^\circ$.  Then, by defining $u_\theta\left(\cdot\right) := D_\theta'\circ \bar{F}_S\left(\cdot\right)$ and $u_{i,\theta}\left(\cdot\right) := D_\theta'\circ \bar{F}_{X_i}\left(\cdot\right)$ for $i=1,2,\dots,n$, the conditions in Proposition \ref{pp:holistic} are fulfilled. 
\end{example}
 
     

\begin{example}  
\label{ex:holistic:elliptic}
   Using the same notations in Example \ref{eg:ellip}, let ${\bf X}\sim E_n(\boldsymbol{\mu},\boldsymbol{\Sigma},g)$ be elliptically distributed, and $Z$ be a spherical random variable which  has the same distribution as $(X_i-\mu_i)/\sigma_i$ for each $i=1,2\dots,n$. Suppose that the conditions in Proposition \ref{pp:holistic} are satisfied, and for any $\theta\in I^\circ$, 
   there exists differentiable distortion functions $(D_{i,\theta}(\cdot))_{i=1}^n$ and $D_\theta(\cdot)$ such that 
   $$\rho_\theta(S) = \mathbb{E}[SD_\theta'(\bar{F}_S(S))] \quad \text{and} \quad  \rho_{i,\theta}(X_i) = \mathbb{E}[X_i D'_{i,\theta}(\bar{F}_{X_i}(X_i))], \ i=1,2,\dots,n. $$
   It is easy to see that 
    \begin{equation*}
        \rho_{i,\theta}(X_i) = \mu_i + \sigma_i \mathbb{E}[ZD'_{i,\theta}(\bar{F}_Z(Z))] =:\mu_i +\sigma_i\lambda_i(\theta), \ i=1,2,\dots,n;
    \end{equation*}
    \begin{equation*}
        \rho_\theta(S) = \mu_S+\sigma_S\mathbb{E}[ZD_\theta'(\bar{F}_Z(Z))] =:\mu_S +\sigma_S\lambda(\theta). 
    \end{equation*}
Hence, $K(\theta)$ is given by, for $I\in\theta$,
    \begin{equation*}
        K(\theta) = \mu_S + (1-\beta)\sigma_S\lambda(\theta) + \beta\sum_{i=1}^n\sigma_i\lambda_i(\theta) =: \mu_S + \Lambda_{\beta,\sigma_S, \boldsymbol{\sigma}}(\theta),
    \end{equation*}
where $\boldsymbol{\sigma}:=(\sigma_i)_{i=1}^n$ and $\Lambda_{\beta,\sigma_S,\boldsymbol{\sigma}} :=  (1-\beta)\sigma_S\lambda(\theta) + \beta \sum_{i=1}^n \sigma_i\lambda_i(\theta)$. 

Using the fact that $\mathcal{S}_{\bf X} \subset K(I)$, thanks to Proposition \ref{pp:holistic}, we can define 
    \begin{equation}
    \label{eq:theta:holistic:ex}
        \Theta := \Lambda_{\beta,\sigma_S, \boldsymbol{\sigma}}^{-1}(S-\mu_S), 
    \end{equation}
such that $K(\Theta)=S$, where $\Lambda_{\beta,\sigma_S, \boldsymbol{\sigma}}^{-1}:\mathbb{R}\to \mathbb{R}$ is a measurable right-inverse of $\Lambda_{\beta,\sigma_S, \boldsymbol{\sigma}}$. Hence, the holistic Pareto-optimal risk-sharing rule  
    \begin{equation*}
        H_i({\bf X}) = K_i(\Theta) = \mu_i+\sigma_i\lambda_i(\Theta) - \frac{\beta_i}{\beta}\left(S-\mu_S - \sigma_S\lambda(\Theta) \right).
    \end{equation*}
    
In particular, if the distortion functions are chosen to be the same, i.e., $D_{i,\theta}\equiv D_\theta$ for $i=1,2,\dots,n$ and $\theta\in I^\circ$, so that $\lambda_i\equiv \lambda$, then, \eqref{eq:theta:holistic:ex} is reduced to 
    \begin{equation*}
        \Theta = \lambda^{-1}\left( \frac{S-\mu_S}{(1-\beta)\sigma_S +\beta\sum_{i=1}^n\sigma_i} \right).
    \end{equation*}
 Hence, the induced risk-sharing rule is reduced to, for $i=1,\dots,n$, 
            \begin{align}
            \label{eq:elliptical:holistic}
                H_i({\bf X}) = K_i(\Theta) &= \mu_i + \frac{\sigma_i -\beta_i\left( \sum_{j=1}^n\sigma_j-\sigma_S \right) }{(1-\beta)\sigma_S + \beta \sum_{j=1}^n\sigma_j}(S-\mu_S) \nonumber \\
                &= \frac{\sigma_i -\beta_i\left( \sum_{j=1}^n\sigma_j-\sigma_S \right) }{(1-\beta)\sigma_S + \beta\sum_{j=1}^n\sigma_j} S + \left( \mu_i -  \frac{\sigma_i -\beta_i\left( \sum_{j=1}^n\sigma_j-\sigma_S \right) }{(1-\beta)\sigma_S + \beta\sum_{j=1}^n\sigma_j} \mu_S\right).
            \end{align}

   Since the aggregate capital in the underlying allocation principle is endogenously solved along with the allocated capitals, the random parameter $\Theta$ in \eqref{eq:theta:holistic:ex} depends on the marginal distributions of ${\bf X}$ via the individual standard deviations $\boldsymbol{\sigma}$,  and the underlying allocation mechanisms via the functions $(\lambda_i\left(\cdot\right))_{i=1}^n$, $\lambda\left(\cdot\right)$, and the parameter $\beta$ which depends on the positive weights $(\gamma_i)_{i=1}^n$, $\gamma$. In the special case when all the distortion functions are identical, we again obtain a quota risk-sharing rule. Compared to the risk-sharing rule based on the Euler principle in Example \ref{ex:distortion:elliptical}, for the $i$-th participant, the participation fraction in \eqref{eq:elliptical:holistic} depends also on the marginal distributions of the others through the sum of the individual standard deviations, which is then weighted via the relative exposures between the aggregate loss and  the individual losses. Nevertheless, it does not depend on the correlation between $S$ and the individual risks as appeared in \eqref{eq:elliptic:distortion}. This is because the holistic approach minimizes an objective function based on additive individual penalties, whereas the Euler principle attributes risk based on each individual's marginal contribution to the aggregated risk $S$, which inherently depends on the covariance term.
   
\end{example}

\section{Managerial Insights}
\label{sec:conclusion}
This article proposed a novel randomization approach to generate risk-sharing rules from capital allocation principles, offering a promising alternative to existing economic and game-theoretic methods. We formulated a general procedure for this randomization approach, and introduced the notion of measurable right-invertibility for the family of aggregate capitals that enables the induction of randomized rules. We demonstrated how risk-sharing rules can be induced under both top-down and bottom-up capital allocation principles.


\subsection{Capital Allocation and Risk Sharing -- Two Sides of the Same Coin}

Under mild conditions, most capital allocation principles (deterministic allocation of capital to business units) can be transformed into a risk-sharing rule (random allocation of actual losses among participants) through a simple randomization procedure. Conversely, many well-known risk-sharing rules can be interpreted as randomized capital allocations.

Organizations that have already been using capital allocation methods, such as banks allocating economic capital to trading desks or insurers allocating capital to business lines, can repurpose these methods to design risk-sharing agreements. For example, a bank that uses VaR-based Euler allocations internally can use the same logic to structure a syndicated loan agreement or a credit risk transfer. This reduces the need to develop separate frameworks for internal capital management and external risk sharing. For instance, a pension fund manager structuring a risk transfer transaction with multiple institutional investors can utilize the same marginal contribution logic from the Euler principle that they use internally to allocate capital, now applied to sharing actual losses among the investors. This ensures philosophical consistency between internal risk management and external risk transfer.

\subsection{Scenario-Wide Property Inheritance and Comonotonicity}

By randomizing different families of capital allocation principles, this paper generated new risk-sharing rules that inherit specific properties from their parent families. Randomizing squared-penalty optimization yields quota-share rules (affine in aggregate loss) when parameters are independent of the randomization parameter, but more complex rules otherwise. Randomizing absolute-penalty optimization yields generalized comonotonic conditional mean risk sharing. Randomizing Euler principles with distortion risk measures yields rules based on marginal contributions evaluated at random confidence levels. Randomizing holistic principles yields rules that balance individual and aggregate penalties.

Managers can ``design'' risk-sharing rules by choosing an appropriate capital allocation family and randomization scheme. If simplicity is paramount (e.g., for regulatory reporting), a quota-share rule from squared penalties may suffice. If fairness in tail events is critical, an Euler-based rule with distortion measures may be appropriate. If both individual and aggregate objectives matter, a holistic principle offers a compromise.

Whenever a property is satisfied by the underlying capital allocation principles, the induced risk-sharing rule inherits this property pathwise. 
For instance, when the capital allocation principles arise from an optimization framework that is Pareto optimal with respect to individual objectives, the induced rule is also Pareto optimal in a scenario-wise sense; see Definition \ref{def:pareto:optimal:pointwise}. Likewise, when the underlying family is derived from the Euler principle, the induced risk-sharing rule yields a mechanism that allocates losses according to each agent’s marginal contribution to the aggregate loss in every scenario. In some cases, this perspective yields new interpretations of familiar rules: the quantile risk-sharing rule, for example, can be viewed as emerging from randomizing optimization-based allocation principles with absolute deviation penalties.

In particular, this method of induction provides a clear mechanism linking the parametrization of capital allocation principles to the comonotonicity of the induced risk-sharing rules. In particular, if the mapping $\theta \mapsto K_i(\theta)$ is simultaneously non-decreasing or non-increasing for all agents $i$, then the resulting risk-sharing rule is necessarily comonotonic.
This connection becomes explicit in several important examples. For top-down allocations based on the Euler principle with distortion risk measures, and for bottom-up weighted risk allocation principles, the induced risk-sharing rules are comonotonic whenever the corresponding CMRS rule is comonotonic; see Propositions~\ref{pp:distortion:euler:comono} and \ref{pp:comono:weighted:risk} respectively.

The following two numerical examples further illustrate such (non-)comonotonicity of the induced risk-sharing rules by the proposed randomization approach; they also offer an additional insight related to dependence structure.

\begin{example}
    \label{ex:numerical:Euler}
    We first provide an example for the implementation of risk-sharing rules induced by the top-down Euler principle in Section \ref{sec:euler_distortion} (distortion risk measures). For $i=1,2$, let $F_i$ denote the distribution function of a Gamma-distributed random variable with shape parameter $\alpha_i$ and scale parameter $\beta_i$. We set $\alpha_1 = 5$, $\beta_1 = 1$, and $\alpha_2 = 0.3$, $\beta_2 = 8$. The first marginal is relatively light-tailed and only moderately right-skewed, whereas the second is highly right-skewed and exhibits substantially heavier right-tail behavior.        

    We consider two risk vectors, ${\bf X}^{(j)} = (X^{(j)}_1,X^{(j)}_2)$, $j=1,2$, whose joint distributions are given by $F_{{\bf X}^{(j)}}(x_1,x_2) = C^{(j)}(F_1(x_1),F_2(x_2))$, for $x_1,x_2\geq 0$, where, for $u,v\in[0,1]$,
        \begin{equation*}
            C^{(1)}(u,v) = (u^{-2} + v^{-2} - 1)^{-\frac{1}{2}}, 
            \qquad 
            C^{(2)}(u,v) = \max\{u+v-1,\,0\}.
        \end{equation*}
That is, ${\bf X}^{(1)}$ is coupled via a Clayton copula, whereas ${\bf X}^{(2)}$ exhibits counter-monotonic dependence. 

Simulating $200{,}000$ observations from the joint distributions, we compute the conditional expectations of $X^{(j)}_i$ given the aggregate loss $S^{(j)} := X^{(j)}_1 + X^{(j)}_2$, for $i,j=1,2$, as shown in Figure~\ref{fig:conditional:expectations}. For ${\bf X}^{(1)}$, the conditional expectations $\mathbb{E}[X^{(1)}_i | S=s]$ are non-decreasing in $s$ for $i=1,2$; this is a consequence of the total positivity of order 2 of ${\bf X}^{(1)}$, which itself follows from the use of the Clayton copula; see, e.g.~Chapter 5.3 in \cite{denuit2006actuarial}.  In contrast, under the counter-monotonic copula, the conditional expectations $\mathbb{E}[X^{(2)}_i | S=s]$ move in opposite directions once $s$ is sufficiently large.
\begin{figure}[!h]
    \centering
     \begin{subfigure}{.5\textwidth}
                \centering
             \includegraphics[scale=0.35]{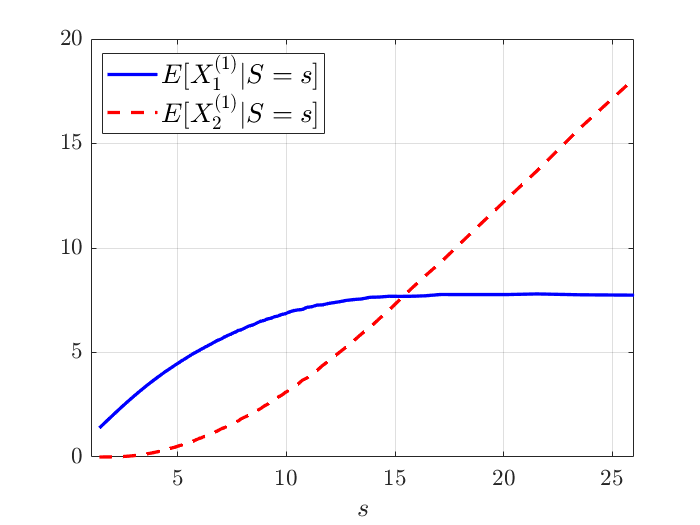}                \caption{$\mathbb{E}[X^{(1)}_i|S^{(1)}=s]$, $i=1,2$}
                \label{fig:cond:exp:1}
            \end{subfigure}%
     \begin{subfigure}{.5\textwidth}
                \centering
             \includegraphics[scale=0.35]{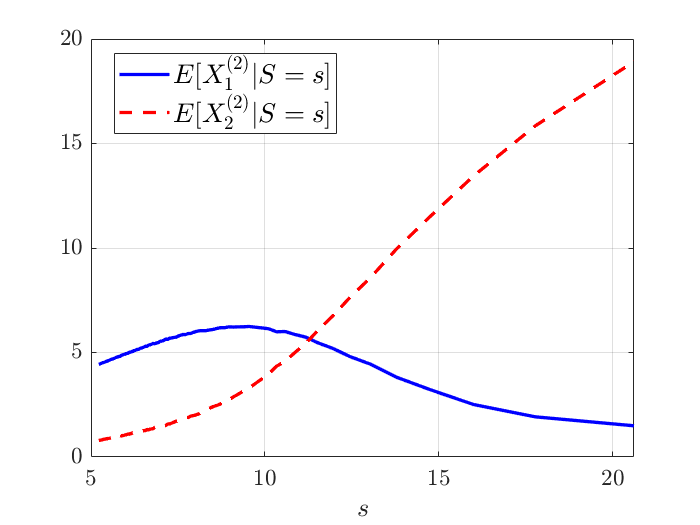}                \caption{$\mathbb{E}[X^{(2)}_i|S^{(1)}=s]$, $i=1,2$}
                \label{fig:cond:exp:2}
            \end{subfigure}%
    \caption{Conditional expectations of individual risks given the aggregate risk}
    \label{fig:conditional:expectations}
\end{figure}

We compute the induced shared risk for each vector using the Euler allocation principle by a distortion risk measure $K^{(j)}(\theta):=\rho_\theta(S^{(j)})$, $j=1,2$, $\theta\in(0,1)$, where the distortion function is given by Wang's transformation
$D_\theta(p) = \Phi(\Phi^{-1}(p) - \Phi^{-1}(1-\theta))$,  $p\in[0,1]$; recall, from Example \ref{eg:wang_example}, that Wang's transformation satisfies the conditions in Proposition \ref{pp:distortion}. The induced risk-sharing rules are given by, for $i,j=1,2$, 
    \begin{equation*}
        H_i({\bf X}^{(j)}) = \frac{\partial}{\partial t}\rho_\theta\left(S^{(j)} + tX^{(j)}_i \right) \bigg|_{t=0}.
    \end{equation*}
Figure \ref{fig:theta:euler} depicts the sampling level $\Theta^{(j)}$, $j=1,2$, against the realization $S=s$, where $\Theta^{(j)} = (K^{(j)})^{-1}(S^{(j)})$. Using this, Figure \ref{fig:H:Euler} displays the induced risk-sharing rules for ${\bf X}^{(j)}$, $j=1,2$, with respect to $S=s$. 
    \begin{figure}[!h]
        \centering
        \includegraphics[width=0.4\linewidth]{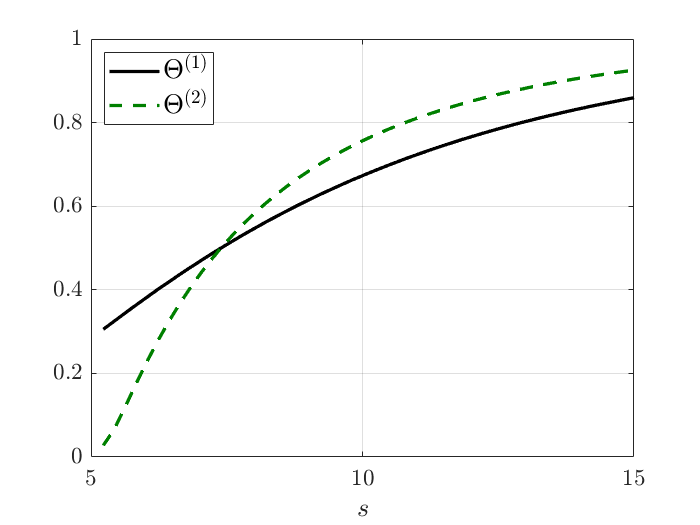}
        \caption{The realized sampling level $\Theta^{(j)}$ with $S^{(j)}=s$, $j=1,2$,}
        \label{fig:theta:euler}
    \end{figure}

Two key observations can be drawn from Figure~\ref{fig:H:Euler}. First, Member~1 (resp.~2) dominates the shared loss when the aggregate $s$ is small (resp.~large). This arises because $X_2$ is more right-skewed and heavy-tailed: for small aggregate losses, the lighter-tailed $X_1$ contributes more due to its higher density near the peak, whereas for large losses, the heavy tail of $X_2$ dominates the aggregate. These reflect the impact of members' marginal contributions as dictated by the Euler principle. Second, the induced risk-sharing rule is comonotonic for ${\bf X}^{(1)}$ but not for ${\bf X}^{(2)}$. Indeed, $\mathbb{E}[X^{(1)}_i|S=s]$ is non-decreasing in $s$ for $i=1,2$, implying comonotonicity of the induced rule (Proposition~\ref{pp:distortion:euler:comono}). In contrast, when $s$ is large, the aggregate loss is primarily contributed by Member~2. Member~1's contribution thus decreases under the counter-monotonic structure. 
\begin{figure}[!h]
    \centering
     \begin{subfigure}{.5\textwidth}
                \centering
             \includegraphics[scale=0.35]{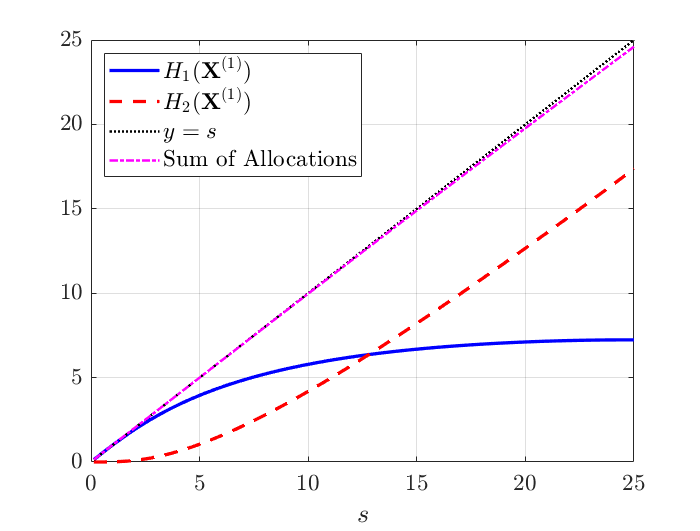}                \caption{${\bf H}({\bf X}^{(1)})$}
                \label{fig:H1:Euler}
            \end{subfigure}%
     \begin{subfigure}{.5\textwidth}
                \centering
             \includegraphics[scale=0.35]{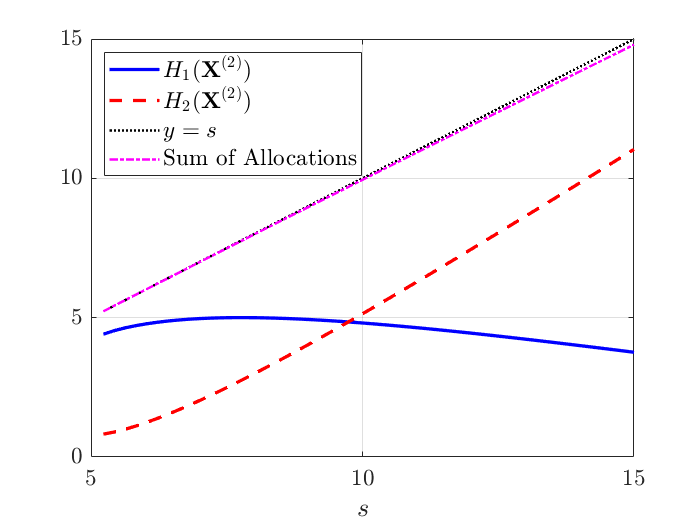}                \caption{${\bf H}({\bf X}^{(2)})$}
                \label{fig:H2:Euler}
            \end{subfigure}%
    \caption{The induced risk sharing rules ${\bf H}({\bf X}^{(j)})$ with $S^{(j)}=s$, $j=1,2$.}
    \label{fig:H:Euler}
\end{figure}

\end{example}

\begin{example}
\label{ex:numerical:weighted:risk}
    We next provide an example for the implementation of the induced risk-sharing rules based on the bottom-up weighted risk allocation principle. We consider the same risk vectors ${\bf X}^{(j)}$, $j=1,2$, as in Example \ref{ex:numerical:Euler}, and choose the size-biased allocation; i.e., $w_\theta(s) = s^\theta$, $s\geq 0, \theta\in \mathbb{R}$. Hence, for $i,j=1,2$, the risk-sharing rule is given by 
        \begin{equation*}
            H_i({\bf X}^{(j)}) = \frac{\mathbb{E}[X^{(j)}_i (S^{(j)})^\theta]}{\mathbb{E}[(S^{(j)})^\theta]}\bigg|_{\theta=\Theta^{(j)}}.
        \end{equation*}
    The sampling level $\Theta^{(j)}$, $j=1,2$, is illustrated in Figure \ref{fig:theta:weighted}.

   \begin{figure}[!h]
        \centering
        \includegraphics[width=0.4\linewidth]{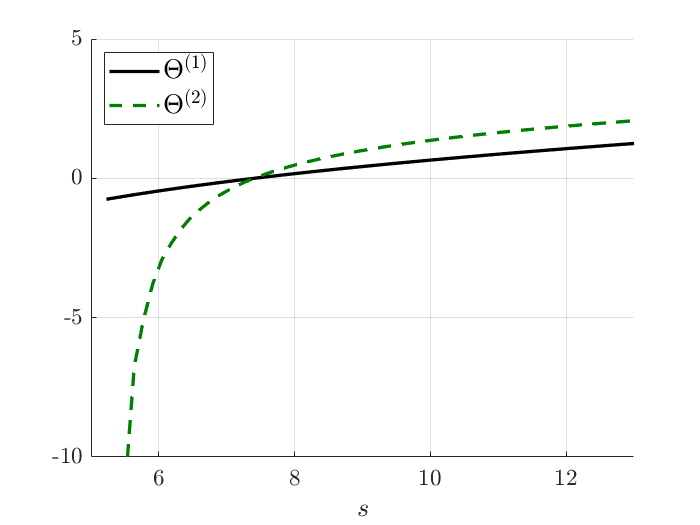}
        \caption{The realized sampling level $\Theta^{(j)}$ with $S^{(j)}=s$, $j=1,2$.}
        \label{fig:theta:weighted}
    \end{figure}

Figure \ref{fig:H:weighted} depicts the shared risk corresponding to the realized aggregate loss $S^{(j)}=s$ for both dependence structures. The observation is similar to that in Example \ref{ex:numerical:Euler}. In particular, we see a comonotonic structure for ${\bf H}({\bf X}^{(1)})$ under the Clayton copula, which arises because $\mathbb{E}[X^{(1)}_i \mid S=s]$ is non-decreasing in $s$ for both $i=1,2$ (Proposition \ref{pp:comono:weighted:risk}).

\begin{figure}[!h]
    \centering
     \begin{subfigure}{.5\textwidth}
                \centering
             \includegraphics[scale=0.35]{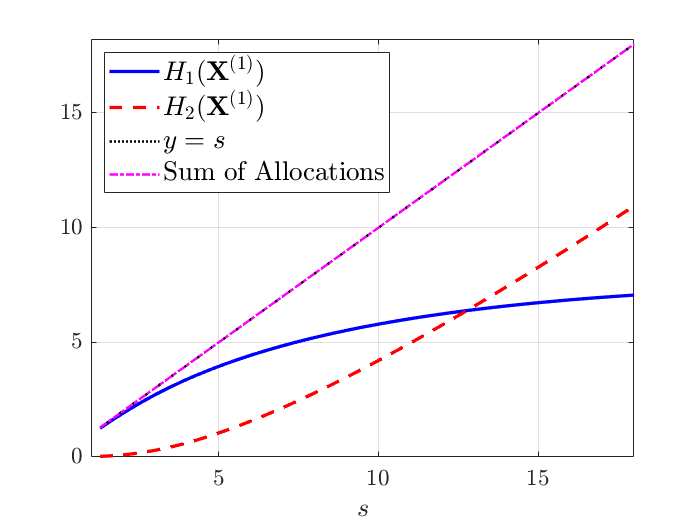}                \caption{${\bf H}({\bf X}^{(1)})$}
                \label{fig:H1:weighted}
            \end{subfigure}%
     \begin{subfigure}{.5\textwidth}
                \centering
             \includegraphics[scale=0.35]{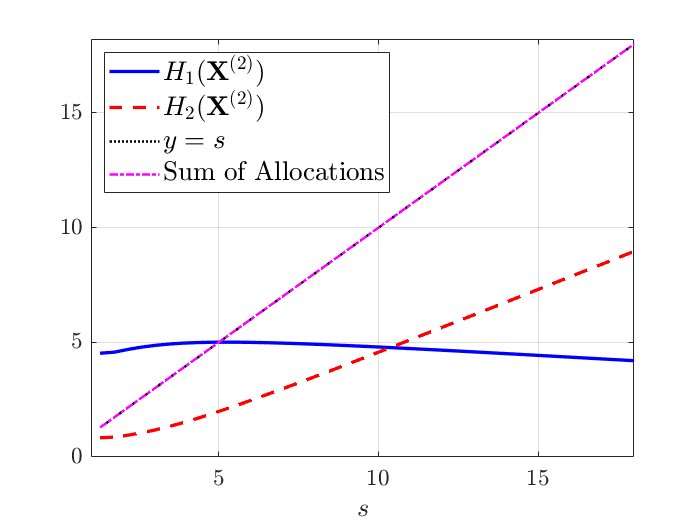}                \caption{${\bf H}({\bf X}^{(2)})$}
                \label{fig:H2:weighted}
            \end{subfigure}%
    \caption{The induced risk sharing rules ${\bf H}({\bf X}^{(j)})$ with $S^{(j)}=s$, $j=1,2$}
    \label{fig:H:weighted}
\end{figure}

\end{example}


\subsection{Dependence Structure Drives Who Bears Tail Risk}
In Examples \ref{ex:numerical:Euler} and \ref{ex:numerical:weighted:risk} with Clayton copula (positive dependence) versus counter-monotonic copula, the induced sharing rules differ dramatically. Under positive dependence, both participants' shares increase with the total loss, and thus are comonotonic. Under counter-monotonic dependence, when the total loss is large, the heavy-tailed participant dominates the aggregate, and the light-tailed participant's share actually decreases, reflecting that a large aggregate loss is driven almost entirely by the heavy-tailed risk.

The choice of risk-sharing rule must account for the dependence structure among participants' risks. If risks are positively correlated, all participants share proportionally in both small and large losses. If risks are negatively correlated or have different tail behavior, the allocation of extreme losses may concentrate on the risk that drives the tail. This insight is critical for designing risk pools where participants have different risk profiles, for example, mixing catastrophe risks with operational risks. As another example, a multinational corporation pooling currency risk (heavy-tailed, more volatile) and operational risk (lighter-tailed, less volatile) across subsidiaries should expect that, in extreme scenarios, the currency risk will dominate, and the sharing rule should reflect this to avoid unfairly burdening stable operations during currency crises.

\subsection{Residual Risk Sharing and Robustness}

The proposed framework opens several avenues for applications and future research. One immediate application is residual risk sharing: a first layer of risks can be allocated according to a given capital allocation principle from a family, and the remaining residual risks can then be shared using the induced risk-sharing rule derived from the same family. This two-stage design ensures philosophical consistency between capital determination and risk redistribution, and it is natural to explore which properties are preserved or modified under this iterative procedure, while also providing a practical mechanism for applications in regulatory and solvency contexts, as well as in collective investment or pension funds. As a second application, the randomized induction framework enables the incorporation of robustness into capital allocations. By generating multiple realized allocations from the same underlying principle, the randomization produces a continuum of plausible capital assignments that capture uncertainty in implementation. This allows one to explore robust or worst-case allocations and to evaluate how variations in the allocation might affect downstream decisions.

\section{Conclusion}

This paper establishes a novel and systematic connection between capital allocation and risk sharing through a randomization approach. By treating a family of capital allocation principles parameterized by a common parameter, and then selecting the parameter level such that the resulting aggregate capital matches the realized aggregate loss, we construct risk‑sharing rules that inherit the structure and properties of the underlying capital allocation principles in a scenario‑wise manner. This framework unifies a wide range of existing risk‑sharing rules, including conditional mean risk sharing, quantile risk sharing, and quota-share rules, as special cases arising from familiar allocation principles such as the Euler principle, optimization‑based allocation, weighted risk allocation, and the holistic principle. Central to the construction is the notion of measurable right‑invertibility of the aggregate capital function, for which we provide sufficient conditions for both top‑down and bottom‑up families. The induced risk‑sharing rules exhibit scenario‑wise Pareto optimality and, under appropriate conditions, yield comonotonic allocations that parallel those obtained from classical economic approaches. Numerical examples illustrate how the dependence structure among participants’ risks shapes the resulting sharing rules, highlighting the practical relevance of the framework for designing risk‑sharing arrangements in contexts ranging from syndicated loans to peer‑to‑peer insurance.

Several directions for future research emerge naturally from this work. First, the randomization framework can be extended to dynamic settings, where capital allocation principles are applied over multiple periods and risk‑sharing rules are updated as information unfolds. Second, the framework opens the door to studying residual risk sharing: after a first‑layer allocation according to a deterministic capital allocation principle, the remaining risks can be shared via the induced rule, and it would be of interest to investigate how properties such as Pareto optimality or comonotonicity are preserved or transformed under such iterative procedures. Third, the inherent randomness in the parameter sampling rule can be exploited to incorporate robustness and model uncertainty; one could quantify the sensitivity of the induced sharing rule to the choice of the underlying family. Finally, empirical applications, such as designing risk‑sharing agreements for insurance pools, infrastructure projects, or financial consortia, would provide valuable practical insights and test the theoretical predictions regarding the role of dependence structures and tail behavior in determining the distribution of shared losses.

\bibliographystyle{plainnat}
\bibliography{ref}

 \begin{appendices}

\section{Proofs of Proposition \ref{pp:K:surjective}}
\label{sec:app:pf:PP31}
  The statement (i) is by definition. The fact that $K(\Theta)=S$ a.s., implies that, for a.a. $\omega\in \Omega$, $K(\Theta)(\omega)=S(\omega)$, and thus there exists $\theta\in I$ such that $S(\omega)=K(\Theta)(\omega)=K\left(\theta\right)$, implying that $S(\omega)\in K\left(I\right)$.

    Suppose that $K(\cdot)$ is measurable, as well as non-decreasing or continuous, and $\mathcal{S}_{{\bf X}}\subseteq K(I)$. For any $s\in \mathcal{S}_{{\bf X}}$,
    \begin{equation*}
    \psi_K(s) := \{\theta\in I : K(\theta) = s \},
    \end{equation*}
    is non-empty and closed. By the measurability of $K(\cdot)$, for any open subset $U\subseteq I$, the set $\{s\in \mathcal{S}_{\bf X} : \psi_K(s)\cap U \neq \varnothing \}$ is Borel measurable.
    By the Kuratowski-Ryll-Nardzewski measurable selection theorem, there exists a measurable selection of the right-inverse function, $K^{-1+}: \mathcal{S}_{{\bf X}} \to I$, such that $K^{-1+}(s) \in \psi_K(s)$, for any $s\in \mathcal{S}_{\bf X}$. The statement (ii) then follows by defining the random variable  $\Theta: \Omega \to I$ by, for any $\omega\in\Omega$, $\Theta(\omega) = K^{-1+}(S)(\omega)$.

\section{Proofs in Section \ref{sec:top:down}}
This appendix consists of proofs of the statements in Section \ref{sec:top:down}.

\subsection{Proof of Proposition \ref{eq:absolute_penalty_sufficient_con}}
\label{sec:app:pf:PP41}
 By the equivalence of $\mathbb{Q}_i^\theta$ and $\mathbb{P}$ for any $\theta\in I$ and $i=1,2,\dots,n$,  we have $ \{ x\in\mathbb{R}: \mathbb{Q}_i^\theta(X_i\leq x) = 0 \} = \{ x\in\mathbb{R} : \mathbb{P}(X_i\leq x) = 0 \}  $, and $\{x\in\mathbb{R}:\mathbb{Q}_i^\theta(X_i>x)=0 \} = \{ x\in\mathbb{R} : \mathbb{P}(X_i>x) =0 \} $ which implies that $\{x\in\mathbb{R}:\mathbb{Q}_i^\theta(X_i\leq x) \leq 1 \} = \{ x\in\mathbb{R} : \mathbb{P}(X_i\leq x) \leq 1 \} $. Thus, by Theorem 6 of \cite{DHAENE20023},
    \begin{align*}
         F^{-1+}_{\bar{S}^c_\theta }(0) &= \sum_{i=1}^n \left(F^{\mathbb{Q}_i^\theta}_i\right)^{-1+}(0) =  \sum_{i=1}^n F_i^{-1+}(0) \leq F^{-1+}_S(0) , \\
         F^{-1 }_{\bar{S}^c_\theta}(1) &= \sum_{i=1}^n \left(F^{\mathbb{Q}_i^\theta}_i\right)^{-1}(1) =  \sum_{i=1}^n F_i^{-1}(1) \geq F_S^{-1}(1). 
    \end{align*}
  Therefore, $\mathcal{S}_{\bf X} \subset [ F^{-1+}_{\bar{S}^c_\theta }(0), F^{-1}_{\bar{S}^c_\theta }(1)]$ for any $\theta\in I$. 
  
  Finally, we have $K(\Theta)\in \mathcal{S}_{{\bf X}} \subset [ F^{-1+}_{\bar{S}^c_\theta }(0), F^{-1}_{\bar{S}^c_\theta }(1)]$, and thus
        \begin{equation*}
            \sum_{i=1}^n H_i({\bf X}) = \sum_{i=1}^n \left(F^{\mathbb{Q}_i^\theta }_i\right)^{-1(\alpha_{K(\theta)})}(F_{\bar{S}^c_\theta}(K(\theta))) \bigg|_{\theta=\Theta}  = F_{\bar{S}^c_\theta}^{-1(\alpha_{K(\theta)} )}(F_{\bar{S}^c_\theta}(K(\theta))) \bigg|_{\theta=\Theta} =K(\Theta) =S.
        \end{equation*}

\subsection{Proof of Proposition \ref{pp:distortion}}
\label{sec:app:pf:PP42}
  We first show that $K(\cdot)$ is non-decreasing.  
            Since $\theta\mapsto D_\theta(p)$ is non-decreasing for any fixed $p\in[0,1]$, one has
            $D_{\theta_2}(\bar{F}_S(s)) \geq D_{\theta_1}(\bar{F}_S(s))$ for any $s\in \mathbb{R}$, where $\theta_2\geq\theta_1$; by \eqref{eq:distortion}, we immediately have $K(\theta_2)\geq K(\theta_1)$. The monotonicity  also implies that $a^*\leq b^*$.

                
            We proceed to show the continuity of $K(\cdot)$ on $(a^*,b^*)$. Since $K(\theta)$ is non-decreasing, by the definition of $a^*$ and $b^*$, we have $|K(\theta)|<\infty$ for $\theta\in (a^*,b^*)$.  Fix $\theta_0 \in (a^*,b^*)$, and consider an increasing sequence $(\theta_n)_{n=1}^\infty$ with $\theta_n\uparrow \theta_0$ when $n\rightarrow\infty$.  For any $s\in\mathbb{R}$, the continuity and monotonicity of $\theta\mapsto D_\theta(\bar{F}_S(s))$ implies 
         $           D_{\theta_n}(\bar{F}_S(s)) \uparrow D_{\theta_0}(\bar{F}_S(s))$ when $n\rightarrow\infty$. Therefore, by the Monotone Convergence Theorem, we have 
            \begin{equation*}
            \lim_{n\to\infty} \int_0^\infty D_{\theta_n}(\bar{F}_S(s))\, ds = \int_0^\infty D_{\theta_0}(\bar{F}_S(s))\, ds. 
            \end{equation*}
        Also, for any $s\in\mathbb{R}$, $\tilde{D}_{\theta_n}(\bar{F}_S(s)) := 1-D_{\theta_n}(\bar{F}_S(s))\downarrow \tilde{D}_{\theta_0}(\bar{F}_S(s))$ when $n\rightarrow\infty$, and $\int_{-\infty}^0 \tilde{D}_{\theta_1}(\bar{F}_S(s))\,ds < \infty$, we can apply the Monotone Convergence Theorem to $\tilde{D}_{\theta_1}(\bar{F}_S(s)) - \tilde{D}_{\theta_n}(\bar{F}_S(s))\geq 0$, to show that 
            \begin{equation*}
                 \lim_{n\to\infty} \int_{-\infty}^0 \left(1-D_{\theta_n}(\bar{F}_S(s))\right)\, ds  = \int_{-\infty}^0 \left(1-D_{\theta_0}(\bar{F}_S(s))\right)\, ds. 
            \end{equation*}
    Combining the two limits, we have $\lim_{\theta\uparrow \theta_0 }K(\theta)=K(\theta_0)$ proving the left-continuity. The right-continuity can be shown in a similar fashion. 
         
            
           Next, we show that $\lim_{\theta \uparrow b}\rho_\theta(S) = F_S^{-1}(1)$. As $\theta\uparrow b$, for any $s\in\mathbb{R}$,
            \begin{equation*}
                 D_\theta(\bar{F}_S(s)) \uparrow \mathbbm{1}_{(0,1]}(\bar{F}_S(s)) =  \mathbbm{1}_{[0,1)}(F_S(s)) = \mathbbm{1}_{(-\infty,F_S^{-1}(1))}(s).
            \end{equation*}
        Hence, by the Monotone Convergence Theorem, 
            \begin{equation*}
                \lim_{\theta\uparrow b}K(\theta) = -\int_{-\infty}^0 \mathbbm{1}_{[F_S^{-1}(1),\infty)}(s)\,ds + \int_0^\infty \mathbbm{1}_{(-\infty,F_S^{-1}(1))}(s)\,ds = F_S^{-1}(1). 
            \end{equation*}
        Likewise, by noting that, for any $s\in\mathbb{R}$, 
     \begin{equation*}
                   D_\theta(\bar{F}_S(s)) \downarrow  \mathbbm{1}_{\{1\}}(\bar{F}_S(s)) = \mathbbm{1}_{\{0\}}(F_S(s)) =\begin{dcases}
                       \mathbbm{1}_{ (-\infty,F^{-1+}_S(0)] }(s), &\text{if }F_S\left(F^{-1+}_S(0)\right)=0, \\
                       \mathbbm{1}_{(-\infty,F_S^{-1+}(0))}(s), &\text{if }F_S\left(F^{-1+}_S(0)\right)>0,
                    \end{dcases} 
                \end{equation*}
       as $\theta \downarrow a$, by the Monotone Convergence Theorem, 
            \begin{equation*}
            \begin{aligned}
                 \lim_{\theta\downarrow a}K(\theta) &= \begin{dcases}
                    -\int_{-\infty}^0 \mathbbm{1}_{(F_S^{-1+}(0),\infty)}(s)\,ds + \int_0^\infty \mathbbm{1}_{(-\infty,F_S^{-1+}(0)]}(s)\,ds &\text{if }F_S\left(F^{-1+}_S(0)\right)=0, \\
                     -\int_{-\infty}^0 \mathbbm{1}_{[F_S^{-1+}(0),\infty)}(s)\,ds + \int_0^\infty \mathbbm{1}_{(-\infty,F_S^{-1+}(0))}(s)\,ds &\text{if }F_S\left(F^{-1+}_S(0)\right)>0,
                \end{dcases}\\
                &= F_S^{-1+}(0). 
            \end{aligned}
            \end{equation*}

        Finally, by the Intermediate Value Theorem, we have $K((a^*,b^*)) = (\lim_{\theta\downarrow a^*}K(\theta),$ $\lim_{\theta\uparrow b^*}K(\theta))$. Combining this with the limits at the boundaries (see also the footnote 1 for the extended limits at $a$ and $b$), we conclude that $K(\cdot)$ is surjective onto $(F_S^{-1+}(0),F_S^{-1}(1))$.

\subsection{Proof of Proposition \ref{pp:distortion:euler:comono}}
\label{sec:app:pf:PP43}
      
      By the continuity of the distribution of ${\bf X}$ and the statement's assumption, the mapping $p\in(0,1) \mapsto \mathbb{E}[X_i|S=\text{VaR}_{p}(S)]$ is (left-)continuous
      and non-decreasing, which is thus the quantile function of some random variable $\tilde{X}_i$, i.e., 
         $\text{VaR}_{p}(\tilde{X}_i)=  \mathbb{E}[X_i|S=\text{VaR}_{p}(S)]$.
    Hence, by \eqref{eq:allocation:distortion}, for any $\theta\in I^\circ$ and $i=1,\dots,n$,
        \begin{equation*}
               K_i(\theta)=  \int_0^1 \text{VaR}_{1-p}(\tilde{X}_i) \,dD_\theta(p) = \rho_\theta(\tilde{X}_i).
        \end{equation*}
    By Proposition \ref{pp:distortion}, $\theta\mapsto K_i(\theta)=\rho_\theta(\tilde{X}_i)$ is non-decreasing for all $i=1,\dots,n$. Therefore, any induced risk-sharing rule ${\bf K}(\Theta) = (K_i(\Theta))_{i=1}^n$ is comonotonic. 

\section{Proofs in Section \ref{sec:bottom-up}}
This appendix consists of proofs of the statements in Section \ref{sec:bottom-up}.

\subsection{Proof of Lemma \ref{lem:K:monotonic:weighted}}
\label{sec:app:pf:lem51}
 For any $\theta_1,\theta_2 \in I^\circ$ with $\theta_2\geq \theta_1$, we have 
        \begin{align*}
             \frac{\mathbb{E}[g(S)w_{\theta_2}(S)]}{\mathbb{E}[w_{\theta_2}(S)]}- \frac{\mathbb{E}[g(S)w_{\theta_1}(S)]}{\mathbb{E}[w_{\theta_1}(S)]} =  \frac{ \mathbb{E}[g(S)w_{\theta_2}(S)]\mathbb{E}[w_{\theta_1}(S)] - \mathbb{E}[g(S)w_{\theta_1}(S)]\mathbb{E}[w_{\theta_2}(S)]}{\mathbb{E}[w_{\theta_2}(S)]\mathbb{E}[w_{\theta_1}(S)]}.
        \end{align*}
  Since $\mathbb{E}[w_{\theta_2}(S)],\mathbb{E}[w_{\theta_1}(S)]>0$,  it suffices to prove the numerator is non-negative. Indeed,  
        \begin{align*}
          & \ \ \ \   \mathbb{E}[g(S)w_{\theta_2}(S)]\mathbb{E}[w_{\theta_1}(S)] - \mathbb{E}[g(S)w_{\theta_1}(S)]\mathbb{E}[w_{\theta_2}(S)]  \\
          &= \int_{\mathcal{S}_{{\bf X}}} \int_{ \mathcal{S}_{{\bf X}}} \left[ g(s_2)w_{\theta_2}(s_2)   w_{\theta_1}(s_1) -  g(s_1)w_{\theta_1}(s_1)   w_{\theta_2}(s_2) \right] \, dF_S(s_2)\,dF_S(s_1) \\
          &=  \int_{\mathcal{S}_{{\bf X}}} \int_{ \mathcal{S}_{{\bf X}}}   w_{\theta_1}(s_1) w_{\theta_2}(s_2) [g(s_2)-g(s_1)] \, dF_S(s_2)\,dF_S(s_1) \\
          &=  \int_{\mathcal{S}_{{\bf X}}} \int_{ \mathcal{S}_{{\bf X}}} w_{\theta_1}(s_2)w_{\theta_2}(s_1)  [g(s_1)-g(s_2)] \,dF_S(s_1)\,dF_S(s_2).
        \end{align*}
    Hence, we can write 
        \begin{align*}
             & \ \ \ \ \mathbb{E}[g(S)w_{\theta_2}(S)]\mathbb{E}[w_{\theta_1}(S)] - \mathbb{E}[g(S)w_{\theta_1}(S)]\mathbb{E}[w_{\theta_2}(S)] \\
             &= \frac{1}{2}\int_{\mathcal{S}_{{\bf X}}} \int_{ \mathcal{S}_{{\bf X}}}\left[w_{\theta_1}(s_1) w_{\theta_2}(s_2)-  w_{\theta_1}(s_2) w_{\theta_2}(s_1) \right] [g(s_2)-g(s_1)] \, dF_S(s_2)\,dF_S(s_1) \\
             &= \frac{1}{2}\int_{\mathcal{S}_{{\bf X}}} \int_{ \mathcal{S}_{{\bf X}}}w_{\theta_1}(s_2)w_{\theta_1}(s_1) \left[\frac{w_{\theta_2}(s_2)}{w_{\theta_1}(s_2)} - \frac{w_{\theta_2}(s_1)}{w_{\theta_1}(s_1)} \right] [g(s_2)-g(s_1)]\,  dF_S(s_2)dF_S(s_1) \\
             &\geq 0,
        \end{align*}
since $s\mapsto w_{\theta_2}(s)/w_{\theta_1}(s)$ and $s\mapsto g(s)$ are non-decreasing. 

\subsection{Proof of Proposition \ref{pp:weighted:surjective}}
\label{sec:app:pf:PP51}

 We first show that $\lim_{\theta\uparrow b}K(\theta) = F_S^{-1}(1)$. It is clear that $K(\theta)\leq F_S^{-1}(1)$ for any $\theta\in (a,b)$. In addition, for any $s\in(F_S^{-1+}(0),F_S^{-1}(1))$, we have 
          \begin{equation*}
          K(\theta)=  \frac{\mathbb{E}[Sw_\theta(S)]}{\mathbb{E}[w_\theta(S)]}   \geq \frac{s \mathbb{E}[w_\theta(S)\mathbbm{1}_{\{S > s \} } ] }{\mathbb{E}[w_\theta(S)\mathbbm{1}_{\{S \leq s \} } ] + \mathbb{E}[w_\theta(S)\mathbbm{1}_{\{S > s \} } ]}. 
        \end{equation*}
    By \eqref{eq:cond:w:limit}, we have $\lim\inf_{\theta\uparrow b} K(\theta) \geq s$; since such $s$ is arbitrary, and by Lemma \ref{lem:K:monotonic:weighted}, $\theta\mapsto K(\theta)$ is non-decreasing, we have $\lim_{\theta\uparrow b}K(\theta) = F_S^{-1}(1)$.

    Next, we show that $\lim_{\theta\downarrow a}K(\theta) = F_S^{-1+}(0)$. First, for any $\theta_1,\theta_2 \in (a,b^*)$ with $\theta_1\leq \theta_2$, using \eqref{eq:cond:K:weighted:w} and the fact that $S\geq 0$, we have, for any $s\in(F_S^{-1+}(0),F_S^{-1}(1))$,
    \begin{align*}
             \frac{\mathbb{E}[Sw_{\theta_1}(S)\mathbbm{1}_{\{S > s \} } ]}{\mathbb{E}[w_{\theta_1}(S)\mathbbm{1}_{\{S \leq s \} }]} &=  \frac{ \mathbb{E}\left[S\frac{w_{\theta_1}(S)}{w_{\theta_1}(s)}  \mathbbm{1}_{\{S > s \} } \right]}{\mathbb{E}\left[\frac{w_{\theta_1}(S)}{w_{\theta_1}(s)}\mathbbm{1}_{\{S \leq  s \} }\right]} \leq \frac{ \mathbb{E}\left[S\frac{w_{\theta_2}(S)}{w_{\theta_2}(s)}  \mathbbm{1}_{\{S > s \} } \right]}{\mathbb{E}\left[\frac{w_{\theta_2}(S)}{w_{\theta_2}(s)}\mathbbm{1}_{\{S \leq s \} }\right]} = \frac{\mathbb{E}[Sw_{\theta_2}(S)\mathbbm{1}_{\{S > s \} } ]}{\mathbb{E}[w_{\theta_2}(S)\mathbbm{1}_{\{S \leq s \} }]}.
        \end{align*}
    Hence, for any $\theta^* \in (a,b^*)$ and any $\varepsilon>0$, we can find $N^* = N^*({\theta^*},\varepsilon)$ such that 
        \begin{equation*}
             \frac{\mathbb{E}[Sw_{\theta}(S)\mathbbm{1}_{\{S > N^* \} } ]}{\mathbb{E}[w_{\theta}(S)\mathbbm{1}_{\{S \leq N^* \} }]} < \varepsilon
        \end{equation*}
    whenever $\theta \leq \theta^*$.

    To proceed, note that it is clear that $K(\theta)\geq F_S^{-1+}(0)$. Let $q_\theta(s) := \frac{\mathbb{E}[w_\theta(S)\mathbbm{1}_{\{S > s \} } ]}{\mathbb{E}[w_\theta(S)\mathbbm{1}_{\{S \leq  s \} } ]}$, for any $s\in(F_S^{-1+}(0),F_S^{-1}(1))$. Then, for any $\theta \in (a, \theta^*)$ with $\theta^* < b^*$, any $s > F_S^{-1+}(0)\geq 0$, any $\varepsilon>0$, and the corresponding $N^*$ chosen as above, we have
         \begin{align*}
              K(\theta) &  \leq \frac{s \mathbb{E}[w_\theta(S)\mathbbm{1}_{\{S \leq  s \} } ] }{\mathbb{E}[w_\theta(S)\mathbbm{1}_{\{S \leq s \} } ] + \mathbb{E}[w_\theta(S)\mathbbm{1}_{\{S > s \} } ]} + \frac{\mathbb{E}[Sw_\theta(S)\mathbbm{1}_{\{S>s \}}]}{\mathbb{E}[w_\theta(S)] } \\
              &\leq \frac{s \mathbb{E}[w_\theta(S)\mathbbm{1}_{\{S \leq  s \} } ] }{\mathbb{E}[w_\theta(S)\mathbbm{1}_{\{S \leq s \} } ] + \mathbb{E}[w_\theta(S)\mathbbm{1}_{\{S > s \} } ]} + \frac{\mathbb{E}[Sw_\theta(S)\mathbbm{1}_{\{ s < S \leq s+N^* \}}]}{\mathbb{E}[w_\theta(S)] }  + \frac{\mathbb{E}[Sw_\theta(S)\mathbbm{1}_{\{S>s+N^* \}}]}{\mathbb{E}[w_\theta(S)] } \\
              &\leq  \frac{s}{1 + q_\theta(s) } +  \frac{(s+N^*)q_\theta(s) }{1 + q_\theta(s) } +  \frac{\mathbb{E}[Sw_{\theta}(S)\mathbbm{1}_{\{S >s+ N^* \} } ]}{\mathbb{E}[w_{\theta}(S)\mathbbm{1}_{\{S \leq s+N^* \} }]} \frac{1}{1+q_\theta(s+N^*) } \\
              &\leq  \frac{s}{1 + q_\theta(s) } +  \frac{(s+N^*)q_\theta(s) }{1 + q_\theta(s) } +  \frac{\mathbb{E}[Sw_{\theta}(S)\mathbbm{1}_{\{S >N^* \} } ]}{\mathbb{E}[w_{\theta}(S)\mathbbm{1}_{\{S \leq N^* \} }]} \frac{1}{1+q_\theta(s+N^*) }\\
              &\leq  \frac{s}{1 + q_\theta(s) } +  \frac{(s+N^*)q_\theta(s) }{1 + q_\theta(s) } +   \frac{\varepsilon}{1+q_\theta(s+N^*) }. 
        \end{align*}
     Using \eqref{eq:cond:w:limit}, we have $\limsup_{\theta\to a^+} K(\theta) \leq s +\varepsilon$. Since $s>F_{S}^{-1+}(0)$ and $\varepsilon>0$ are arbitrary, and $K(\cdot)$ is non-decreasing, we conclude that     $\lim_{\theta\downarrow a}K(\theta) = F_S^{-1+}(0)$. 
        

  Finally, we show that $K : (a,b^*) \to \mathcal{S}_{{\bf X}}$ is continuous. The surjectivity then follows from the Intermediate Value Theorem and the limits of $K(\cdot)$ at the two end-points. Indeed, for any $\theta_0 \in (a,b^*)$ and any sequence $\{\theta_n\}_{n=1}^{\infty}\subseteq (a,b^*)$ with $\theta_n \to \theta_0$,  using the fact that $\theta\mapsto w_\theta(s)$ is continuous, by \eqref{eq:cond:w:integrability}, and by the Dominated Convergence Theorem applying to the numerator and denominator, 
\[
K(\theta_n) = \frac{\mathbb{E}[S w_{\theta_n}(S)]}{\mathbb{E}[w_{\theta_n}(S)]} \to \frac{\mathbb{E}[S w_{\theta_0}(S)]}{\mathbb{E}[w_{\theta_0}(S)]} = K(\theta_0),
\]
and thus $K(\cdot)$ is continuous on $(a,b^*)$.

\subsection{Proof of Proposition \ref{pp:comono:weighted:risk}}
\label{sec:app:pf:PP52}
 For any $i=1,\dots,n$ and $\theta\in I$, we have 
        \begin{equation*}
            K_i(\theta) = \frac{\mathbb{E}[X_i w_\theta(S)]}{\mathbb{E}[w_\theta(S)]} = \frac{\mathbb{E}[g_i(S)w_\theta(S)]}{\mathbb{E}[w_\theta(S)]},
        \end{equation*}
    where $g_i(s) := \mathbb{E}[X_i|S=s]$. Since $g_i(\cdot)$ is non-decreasing, by Lemma \ref{lem:K:monotonic:weighted}, $K_i(\cdot)$ is also non-decreasing for all $i=1,\dots,n$, and the comonotonicity of \eqref{eq:weighted:risk:sharing} thus follows.  

\subsection{Proof of Proposition \ref{pp:holistic}}
\label{sec:app:pf:PP53}
  By \eqref{eq:holistic:sol}, for $\theta\in I$,
            \begin{equation*}
                K(\theta) = (1-\beta)\rho_\theta(S) + \beta \sum_{j=1}^n \rho_{j,\theta}(X_j).
            \end{equation*}
   Using the fact that $0\leq \beta\leq 1$,
        \begin{align*}
           \lim_{\theta\to a^+} K(\theta) &\leq (1-\beta) F_S^{-1+}(0) + \beta \sum_{j=1}^n F_{X_i}^{-1+}(0) \\
           &= (1-\beta) F_S^{-1+}(0) + \beta F_{S^c}^{-1+}(0) \leq F_S^{-1+}(0);\\
            \lim_{\theta\to b^-} K(\theta) &\geq (1-\beta) F_S^{-1}(1) + \beta \sum_{j=1}^n F_{X_i}^{-1}(1) \\
           &= (1-\beta) F_S^{-1}(1) + \beta F_{S^c}^{-1}(1) \geq F_S^{-1}(1),
        \end{align*}
    where $S^c$ is the comonotonic sum of $(X_i)_{i=1}^n$.  Since $K(\cdot)$ is continuous on $(a,b)\subset I$, we infer that  $(F_S^{-1+}(0),$ $F_S^{-1}(1))\subset K((a,b))$. By Proposition \ref{pp:K:surjective}, there exists a random variable $\Theta$ taking values in $I$ such that $K(\Theta)=S$, and thus \eqref{eq:holistic:risk} is a risk-sharing rule.

\end{appendices}

\end{document}